\documentclass[a4paper,reqno]{amsart}
\usepackage{geometry}
\usepackage{fullpage}
\usepackage[latin1]{inputenc}
\usepackage[T1]{fontenc}
\usepackage{amsfonts}
\usepackage{amssymb}
\usepackage{amsmath}
\usepackage{amsthm}
\usepackage{graphicx}
\usepackage[dvipsnames]{xcolor}
\usepackage{afterpage}
\usepackage[colorlinks=true, linkcolor=magenta, citecolor=cyan, urlcolor=green]{hyperref} 
\usepackage{bbm}
\usepackage{latexsym}

\newcommand{\si}{\sigma_i}
\newcommand{\sj}{\sigma_j}

\newcommand{\xiimu}{\xi_i^\mu}
\newcommand{\xijmu}{\xi_j^\mu}

\newcommand{\summup}{\sum_{\mu=1}^{P}}

\newtheorem{lemma}{Lemma}

\newtheorem{thm}{Theorem}

\numberwithin{equation}{section}

\definecolor{light}{gray}{.9}

\def\be{\begin{equation}}
\def\ee{\end{equation}}
\def\bea{\begin{eqnarray}}
\def\eea{\end{eqnarray}}

\def\s{\sigma}

\def\a{\alpha}
\def\e{\varepsilon}
\def\epsilon{\e}

\def\b{\beta}

\newcommand{\meanv}[1]{\left\langle#1\right\rangle}

\newcommand{\nocontentsline}[3]{}
\newcommand{\tocless}[2]{\bgroup\let\addcontentsline=\nocontentsline#1{#2}\egroup}

\DeclareMathSymbol{\leqslant}{\mathalpha}{AMSa}{"36} 
\DeclareMathSymbol{\geqslant}{\mathalpha}{AMSa}{"3E} 
\DeclareMathSymbol{\eset}{\mathalpha}{AMSb}{"3F}     
\renewcommand{\leq}{\;\leqslant\;}                   

\title{Non-Convex Multi-species Hopfield models}
\date{\today}

\author{Elena Agliari}

\address{Elena Agliari: Dipartimento di Matematica, Sapienza Universit\`a di Roma, Italy.}
\email{elena.agliari@mat.uniroma1.it}

\author{Danila Migliozzi}

\address{Danila Migliozzi: Dipartimento di Matematica, Sapienza Universit\`a di Roma, Italy.}
\email{danila.migliozzi@yahoo.it}

\author{Daniele Tantari}

\address{Daniele Tantari: Scuola Normale Superiore, Pisa, Italy.}
\email{daniele.tantari@sns.it}

\begin{document}

\begin{abstract}
In this work we introduce a multi-species generalization of the Hopfield model for associative memory, where neurons are divided into groups and both inter-groups and intra-groups pair-wise interactions are considered, with different intensities. Thus, this system contains two of the main ingredients of modern Deep neural network architectures:  Hebbian interactions to store patterns of information and multiple layers coding different levels of correlations. The model is completely solvable in the low-load regime with a suitable generalization of the Hamilton-Jacobi technique, despite the Hamiltonian can be a non-definite quadratic form of the magnetizations. The family of multi-species Hopfield model includes, as special cases, the 3-layers Restricted Boltzmann Machine (RBM) with Gaussian hidden layer and the Bidirectional Associative Memory (BAM) model.
\end{abstract}

\maketitle

\section{Introduction}

Artificial Intelligence (AI) is nowadays playing a major role in our everyday lives and has grown extremely fast in the past few years, both in popularity and in scientific advances.
%
%
%
The rise over the past decade has been strongly correlated with the increased processing power of modern computers: the progressions have made complex computations more accessible and allowed for AI-models to train on data-sets so large and multifaceted it would not have been possible a few years ago. 

Neural networks are playing an important and renewed role in this trend mainly because of the success of Deep Learning \cite{Bengio,DL-book} in several applications, ranging from engineering and computer science to neuro and computational biology.  From a mathematical perspective a neural network consists of many simple, connected processors (called neurons), each associated to an activity state depending on the state of the other neurons and on a possible external stimulation. In the family of neural network models the Hopfield model plays a major role. It was introduced to mimic the ability of the brain to retrieve information previously stored (see e.g., \cite{AMIT,Ton, ags1, ags2}) but recently it has been investigated also because of its equivalence with the fundamental constituent of a Deep architecture, namely a two-layer neural network called Restricted Boltzmann Machine (RBM) \cite{hinton1, hinton2,RBM1, RBM2,BBCS}. 

This bridge is particularly interesting since, while a solid theoretical description does exist for the mechanisms behind the retrieval capabilities of the Hopfield model, a clear and exhaustive theoretical scaffold for the performances of  Deep Networks is still needed and most of the advances in technology are often achieved by specific home-made recipes, without a global and systematic understanding of the methodology. For example in \cite{BGST1,BGST2,remi, huang, huang2} it has been shown how the statistical mechanics analysis of the Hopfield model phase diagram allows answering some issues concerning both the way RBMs extract features from data and their efficiency in terms of the training-set size necessary for a good generalization.  

The heart  of the Hopfield model is the particular shape of the interaction between neurons that follows the well-known Hebbian rule \cite{Hebb}.  The structural property of a deep architecture is the presence of several layers or groups of units (neurons) where only neurons belonging to two consecutive layers are connected: the aim of each layer is to codify more and more abstract levels of correlations.  In this work we consider an extension of the standard Hopfield model, with multiple groups of neurons  and different Hebbian-like interactions for each couple of groups. The multi-species model considered in this paper contains also,  as special cases corresponding to particular choices of the parameters,  the 3-layers RBM and the Bidirectional Associative Memory (BAM) model. The latter was introduced \cite{kosko1988bidirectional,kurchan1994statistical,englisch1995bam} to mimic the ability of the human brain to retrieve informations through  association of ideas.

Remarkably, in our model interactions between groups can be suitably tuned in such a way that inter-groups interaction may possibly result stronger than the intra-group one.  As a consequence, the model turns out to be in the family of non-convex multi-species systems. It is known that non-convexity can yield to strong difficulties in the statistical mechanics analysis, especially when dealing with disordered systems.  For example, the multi-species Sherrington Kirkpatrick (SK) model for spin glasses \cite{barra2015multi,Bip,barra2014mean}  still presents several mathematical challenges in the non-convex region, while it was completely explored in the convex regime \cite{barra2015multi, panchenko2015free}. Some properties of the multi-species SK model, in particular the overlaps synchronization,  have recently driven towards the investigation of a multipartite version of the Generalized Random Energy model \cite{genovsinc}. In the case of multipartite ferromagnets the analysis is relatively viable \cite{Contucci} and a suitable strategy based on convexification of the problem has been recently introduced \cite{genovese2016non} and used to compute the model free energy, the equilibrium states and the  thermodynamic fluctuations. Hebbian models are disordered systems but they are purely ferromagnetic in the low load regime, see Section \ref{sec:Hopfield}. For this reason in this paper we can generalize the strategy in \cite{genovese2016non} to find an exact solution of the multi-species Hopfield model in that regime.

The remaining of the paper is organized as follows.
In Sec.~\ref{sec:Hopfield} we review the basic Hopfield model, that is the starting point for our work.
In Sec.~\ref{sec:model} we introduce the multi-species Hopfield model and the main observables used for its investigation. We also give the main theorem concerning the variational principle for the free energy and the equations for the optimal parameters. In Sec.~\ref{sec:special} we discuss two particular cases of the model, corresponding to the BAM and the 3-layer RBM.
In Sec.~\ref{sec:solution} we present the proof of the main therorem.  Finally, Sec.~\ref{sec:conclusions} is left for conclusions and outlooks. Some technical details about the solution are collected in the Appendix.

\section{A brief review on the Hopfield model} \label{sec:Hopfield}
Before introducing the extended model considered in this work it is worth reviewing the standard Hopfield model (see also e.g., \cite{AMIT,Ton, article} for a more extensive explanation).

We consider $N$ binary neurons (i.e., Ising spins for statistical physicists \cite{AMIT} or McCulloch-Pitts neurons for computer scientists \cite{MC-P}) and to each neuron $i$ we assign a variable $\si$ that describes its activity: if $\si = +1$ the $i$-th neuron is spiking, while if $\si = -1$ the $i$-th neuron is quiescent. We denote with $\boldsymbol{\sigma} \in \{ -1, +1\}^N$ the overall configuration of the system.
\newline
Neurons are embedded on a fully-connected network in such a way that the internal field $h_i$ acting on the $i$-th neuron is given by
\begin{equation*}
 h_i (\boldsymbol{\sigma}) = \sum_{\substack{j=1 \\ j\neq i}}^N J_{ij} \sj,
\end{equation*}
where $\mathbf{J} = \{J_{ij}\}_{i,j=1,...,N}$ is the \emph{synaptic} coupling between neuron $j$ and neuron $i$. 
Inspired by neurophysiological mechanisms, one can introduce a dynamic rule for the spin configuration as follows: extract randomly and uniformly a spin $i$ and update its state according to 
\begin{equation} \label{eq:dinamica}
\si(t) \rightarrow \si(t+1)  =  \textrm{sgn} \left[ h_i(\boldsymbol{\sigma}(t)) + T z_i(t) \right],
\end{equation}
where $t$ is the discrete time unit, $z$ is a random variable\footnote{The random variable $z$ is chosen symmetrically distributed and, typically, its probability density is taken as $p(z) = [1- \tanh^2(z)]/2$.} and $T := 1/ \beta$ is a measure of the degree of noise within the system, usually referred to as temperature. 
As long as $\boldsymbol{J}$ is symmetric (i.e., $J_{ij}=J_{ji}, \forall i,j$) and devoid of self-interactions (i.e., $J_{ii} =0$, $\forall i$), this dynamic is ergodic and there exists an invariant measure given by 
\begin{equation} \label{eq:peq}
p_{\textrm{eq}} (\beta) = \frac{1}{Z_N(\beta  | \mathbf{J} )} e^{- \beta H_N(\boldsymbol{\sigma} | \mathbf{J})},
\end{equation}
where $Z_N(\beta  | \mathbf{J})$ is a normalization factor, also called partition function, and
\begin{equation} \label{eq:hamiltonian}
H_N(\boldsymbol{\sigma} | \mathbf{J}) = - \sum_{1 \leq i < j \leq N} \sigma_i J_{j} \sigma_j.
\end{equation}
In a statistical-mechanics context 
the distribution (\ref{eq:peq}) corresponds to the Boltzmann-Gibbs measure used to describe the canonical equilibrium of a system described by the Hamiltonian (\ref{eq:hamiltonian}).

Now, the goal of the system is to be able to recognize and retrieve a certain group of words, pixels, or, generically, patterns. A \emph{pattern} is defined as a sequence of random variables $\boldsymbol{\xi}^{\mu} = (\xi_1^{\mu}, \ldots, \xi_N^{\mu})$ with $\xi_i^{\mu} \in \{-1, +1\}$, $\forall i =1, ..., N$, and $\forall \mu=1, ..., P$, where the label $\mu$ distinguishes different patterns. In the following we shall assume the set $\left\{ \xi_i^\mu \right\}_{i,\mu}$ made of i.i.d. random variables such that 
\begin{equation} \label{orto}
\mathbb{P} (\xi_i^{\mu}=+1) = \mathbb{P}(\xi_i^{\mu}=-1) = 1/2, ~~ \forall i, \mu.
\end{equation}
More general patterns distribution can be investigated: correlated patterns \cite{correlated,ABDG,gutfreund1988neural},  diluted patterns, i.e. which may have zero entries \cite{prlnoi1,multi, immhigh, immmed, immlett,agliari2014multitasking}, Gaussian patterns \cite{BGG-JSP2010,barra2012glassy,bg,gaussSK,legendre}, weighted patterns \cite{anergy}. With such premises we can give a definition of \emph{retrieval}: we say that the system is able to retrieve the $\mu$-th pattern if, given a suitable starting point (i.e., a configuration belonging to the attraction basin of $\boldsymbol{\xi}^{\mu}$), the spin configuration $\boldsymbol{\s}$ converges to $\boldsymbol{\xi}^{\mu}$ under the dynamics (\ref{eq:dinamica}). 
 One can prove that a coupling matrix defined according to Hebb's learning rule \cite{Hebb} as
\begin{equation} \label{hebbrule}
 J_{ij} = \frac{1}{N} \summup \xiimu \xijmu \; 
\end{equation}
ensures the attractiveness of the $P$ patterns, as long as the noise $T$ and the  number of patterns $P$ are not too large. For instance, by taking $P$ finite (or, still, sublinear with respect to $N$) in the thermodynamic limit $N \rightarrow \infty$, one has retrieval capabilities as long as $T<1$.  A modification of the Hebb rule results in a deformation of the basins of attractions: this can be done for example overlaying  Hebb to another interaction structure as a diluted \cite{sompodil,wemmedil} or a hierarchical  \cite{prlnoi3, dysnn,dysjpa,dyspre} structure.

In order to describe the overall state of the system, one introduces the macroscopic observable $\mathbf{m}$, also called \emph{Mattis magnetization}, that is a vector of length $P$, whose $\mu$-th component represents the overlap between the spin configuration and the $\mu$-th pattern:
\begin{equation} \label{eq:Mattis}
m_{\mu} (\boldsymbol{\s} | \boldsymbol{\xi}) := m_{\mu} = \frac{1}{N} \sum_{i=1}^N \xi_i^{\mu} \sigma_i ~~~ \in [-1,1].
\end{equation}
Notice that $m_{\mu}$ can also be written as $m_{\mu} = 2 d_H (\boldsymbol{\s} , \boldsymbol{\xi}) - 1$, where $d_H(x,y)$ is the Hamming distance between the strings $x$ and $y$. The Mattis magnetization plays as the \textit{order parameter} for the system as a certain arrangement for $\mathbf{m} = (m_1, m_2,..., m_P)$ can be associated to a global state for the system; in particular, the retrieval of a pattern $\mu$ corresponds to $m_{\mu} \neq 0$. It is worth underlying that, in the thermodynamic limit, (\ref{orto}) yields to mutually orthogonal patterns in such a way that only one pattern at a time can be retrieved exactly, i.e. $m_{\mu}= 1$.
Now, combining (\ref{eq:hamiltonian}), (\ref{hebbrule}) and (\ref{eq:Mattis}) we can state that the Hamiltonian for the Hopfield model equipped with $N$ Ising neurons and $P$ patterns is defined as
\begin{equation} \label{HopfieldH}
 H_N (\boldsymbol{\sigma} |\boldsymbol{\xi} ) = - \frac{1}{N} \sum_{1 \leq i<j \leq N} \sum_{\mu=1}^P  \xiimu \xijmu \si\sj = - N \sum_{\mu=1}^P m_{\mu}^2 + \frac{P}{N}.
\end{equation}
One can also account for an external field $\mathbf{h} = (h_1, h_2, ..., h_P)$ which biases the neuron configurations toward the retrieval of the $\mu$-th pattern with a relative magnitude $h_{\mu}$, leading to
\begin{equation} \label{HopfieldHh}
 H_N (\boldsymbol{\sigma}|\boldsymbol{\xi} ,\mathbf{h}) = - N \sum_{\mu=1}^P m_{\mu}^2 - N \sum_{\mu=1}^P m_{\mu} h_{\mu},
\end{equation}
where we neglected the unnecessary constant term $P/N$. 
The partition function  $Z_N (\beta,\mathbf{h}|\boldsymbol{\xi})$ for the Hopfield network then reads as
\begin{equation} \label{PF}
Z_{N}(\beta,\mathbf{h} |\boldsymbol{\xi}) = 
\sum_{\boldsymbol{\sigma}} \exp \left [ \beta N \left ( \sum_{\mu=1}^P  m_{\mu}^2 + \sum_{\mu=1}^P h_{\mu} m_{\mu} \right ) \right].
\end{equation}
Hereafter, whenever suitable, we will drop the dependence on $\boldsymbol{\xi}$ and on $\mathbf{h}$ in order to lighten the notation.
We also define the expectation of the observable $O(\boldsymbol{\sigma})$, that is a function of the state $\boldsymbol{\sigma}$, as the average with respect to the equilibrium (or Boltzmann-Gibbs) distribution (\ref{eq:peq}):
\begin{equation}
\langle O \rangle_N =\displaystyle\sum_{\boldsymbol{\sigma}} O(\boldsymbol{\sigma}) p_{\textrm{eq}}(\boldsymbol{\sigma}) = \frac{1}{Z_N(\beta)} \displaystyle\sum_{\boldsymbol{\sigma}} O(\boldsymbol{\sigma})e^{-\beta H_N(\boldsymbol{\sigma} )}.
\label{mediaens}
\end{equation}
In particular, in the following, we shall be interested in the expectation $\langle \mathbf{m} \rangle_N$ 
which, in a statistical-mechanics framework, can be obtained by extremizing the intensive pressure (or, equivalently, the free energy\footnote{The intensive pressure, here denoted as $f_N(\beta)$, is strictly related to the (possibly more familiar) free energy $\tilde{f}_N(\beta)$ by $f_N(\beta) = - \beta \tilde{f}_N(\beta)$. Therefore, the existence and the uniqueness of $\tilde{f}_N(\beta)$ also ensure the existence and uniqueness of $f_N(\beta)$ and vice versa, while the positive convexity of $\tilde{f}_N(\beta)$ ensures the negative convexity of $f_N(\beta)$. As a result, the thermodynamic equilibrium can be detected as a minimum for the free energy or as a maximum for the pressure.}) of the model. More precisely, recalling that the intensive pressure at finite size $N$ is defined as 
\begin{equation} \label{effe}
f_N(\beta,\mathbf{h}|\boldsymbol{\xi}):= \frac{1}{N}\ln Z_N(\beta,\mathbf{h}|\boldsymbol{\xi}),
\end{equation} 
and denoting the thermodynamic limit as $f(\beta,\mathbf{h}|\boldsymbol{\xi}) =\lim_{N\to\infty} f_N(\beta,\mathbf{h}|\boldsymbol{\xi})$, it holds 
\be\label{derl}
\langle \mathbf{m} \rangle=\lim_{N\to\infty}\langle \mathbf{m} \rangle_N=\lim_{N\to\infty}\frac 1 \b \nabla_{\boldsymbol{\lambda}}f_N(\beta,\mathbf{h}+\boldsymbol{\lambda}|\boldsymbol{\xi})|_{\boldsymbol{\lambda}=0}
=\frac 1 \b \nabla_{\boldsymbol{\lambda}}f(\beta,\mathbf{h}+\boldsymbol{\lambda}|\boldsymbol{\xi})|_{\boldsymbol{\lambda}=0}.
\ee

For the Hopfield model (\ref{HopfieldHh}), in the low-load regime, the limiting free energy $f(\beta,\mathbf{h}|\boldsymbol{\xi})$ is proved to exist and to be selfaveraging over the patterns' noise. 
Moreover, one can write a.s. 
 \begin{equation}
 f(\beta,\mathbf{h}) =\displaystyle \lim_{N\to\infty} \frac{1}{N}\ln Z_N(\beta,\mathbf{h}|\boldsymbol{\xi})=\displaystyle\sup_{\textbf{M}} \left [ \ln 2+  \left \langle \ln\cosh (\beta \boldsymbol{\xi} ( \mathbf{M} + \mathbf{h} ) ) \right\rangle_{\boldsymbol{\xi}}- \frac{1}{2}\sum_{\mu=1}^P M_\mu^2  \right ],
 \label{presione}
 \end{equation}
%
%
where $\langle \ \cdot \ \rangle_{\boldsymbol{\xi}}$ means the average over the patterns $\boldsymbol{\xi}$.
By extremizing the previous expression for the pressure with respect to the trial magnetization $\textbf{M}$, we get the following \emph{self-consistent} equations
\begin{equation}
M_{\mu}= \langle\xi^{\mu} \tanh [ \beta \boldsymbol{\xi} ( \textbf{M} +\textbf{h}) ]\rangle_{\boldsymbol{\xi}}, ~~~ \mu=1, ..., P,
\label{gaiona}
\end{equation}
or, in vectorial notation,
\begin{equation}
\textbf{M}=\langle \boldsymbol{\xi} \tanh [\beta \boldsymbol{\xi} (\textbf{M} + \mathbf{h}) ]\rangle_{\boldsymbol{\xi}}.
\label{gaiaa}
\end{equation}
Using equations $(\ref{derl})$ and $(\ref{presione})$ it holds that  the solution of (\ref{gaiaa}) coincides with  $ \langle \mathbf{m} \rangle$ and gives the behavior of the magnetizations with respect to the system parameters. In particular, posing $\mathbf{h}=0$, one can see that as $T>1$ the only solution is $M_{\mu}=0$, $\forall \mu=1, ..., P$, that is, noise prevails and the system is not able to retrieve, while as $T<1$ the system exhibits a symmetry breaking: equations (\ref{gaiaa}) have multiple solutions, each one related to a different way of taking the limit $\mathbf{h}\to0$ and there are solutions with  (at least) one  non-null component of the magnetization, where the corresponding pattern is retrieved. At $T=1$ the system exhibits a continuous phase transition.

The phase diagram becomes richer at high load (i.e, when $P$ grows linearly with $N$) and the equation $(\ref{presione})$ does not hold anymore. In this case it is well known the existence of another phase transition from the retrieval phase to a spin-glass phase, where the system can freeze on thermodynamic states that are not correlated with  any pattern. The occurrence of such new phase depends on the number of stored patterns: when $\a=P/N$ exceeds a critical capacity the system is not able to retrieve anymore. The study of the phase diagram and the value of the critical capacity have been intensively investigated in the literature with approximated techniques from statistical physics \cite{Ton,ags1,ags2}. Conversely, rigorous results are sparser and partial \cite{pastur1994replica,Tala1,Tala2,BGP3,BG5,bov-gen}.

\section{The multi-species Hopfield model} \label{sec:model}

In this section we introduce a multi-species Hopfield model which generalizes the standard model introduced in the previous section by allowing for neurons belonging to different groups (or species) characterized by different inter-group and intra-group couplings.\newline
Let $\nu\ge2$ be the number of species, each made of $N_{a}$ neurons, $a= 1, . . . , \nu$. 
We denote with $\sigma_i^a\in\{ -1, 1\}$, $i=1, . . . , N_a$, the state of the $i$-th neuron in the  $a$-th group.
The overall number of neurons is $N= \sum_{a=1}^{\nu} N_{a}$ and, in the thermodynamic limit $N\to\infty$, we define the parameters 
\begin{equation}\label{alfas}
\alpha_a :=\lim_{N\to\infty} \frac{N_a}{N} \in(0,1), ~~ \textrm{with} ~~ a=1, . . . , \nu 
\end{equation}
Without loss of generality we assume always $\alpha_a =N_a/N$, also before the limit. 
We associate to each group $a$ a set of $P$ patterns $\boldsymbol{\xi}^{\mu,a}=\{\xi^{\mu,a}_i\}_{i=1}^{N_a}$, $\mu=1,...,P$, drawn according to (\ref{orto}) and a set of $P$  \textit{Mattis magnetizations} 
\begin{equation} \label{eq:ma}
m_a^\mu:=m_a^\mu(\boldsymbol{\s}|\boldsymbol{\xi})=\frac{1}{N_{a}}\displaystyle\sum_{i=1}^{N_a} \xi_{i}^{\mu,a} \sigma_{i}^a  \;\;\;\;\;\;\;\;\;\;\;\; \forall \mu=1,\dots, P .
\end{equation}

\noindent
In other terms, $m_a^\mu$ represents the normalized overlap between the $\mu$-th \textit{pattern} $\boldsymbol{\xi}^{\mu,a}$, related to the $a$-th group, and the neuron configuration $\boldsymbol{\sigma}^a$ of the same group. 

The Hamiltonian $H_{N,\nu}(\boldsymbol{\sigma})$ describing the multi-species Hopfield model reads as
\begin{align}
\label{bubu}
H_{N,\nu}(\boldsymbol{\sigma}|\boldsymbol{\alpha},\boldsymbol{k},\boldsymbol{\xi}) &=-\frac{1}{2N}\sum_{a=1}^\nu\sum_{i,j=1}^{N_a} k_a \sum_{\mu=1}^P \xi_i^{\mu,a} \xi_j^{\mu,a} \sigma_i^a\sigma_j^a -\frac{1}{N}\sum_{(a, b)}^\nu\sum_{i=1}^{N_a}\sum_{j=1}^{N_b} \sum_{\mu=1}^P\xi_i^{\mu,a}\xi_j^{\mu,b} \sigma_i^a \sigma_j^b
\\
\label{bubu2}
&=-\frac{N}{2}\left(\displaystyle\sum_{\mu=1}^P\sum_{a=1}^{\nu}k_a \alpha_a^2({m_a^\mu})^2 +\displaystyle\sum_{\mu=1}^P \sum_{a, b=1}^{\nu} \alpha_a\alpha_b m_a^\mu m_b^\mu\right),
\end{align}
where the first term is due to interactions between neurons belonging to the same group, while the second term is due to interactions between neurons belonging to different groups; also, $k_a\in(0,1)$, $\forall a\in\{\ {1,\dots,\nu}\}$, are introduced to tune the magnitude of the intra-group respect to the inter-group interactions. Clearly, setting $\nu=1$, the second terms in (\ref{bubu}) and in (\ref{bubu2}) vanish and we recover the standard Hopfield model (\ref{HopfieldH}), i.e. $H_{N,1} (\boldsymbol{\sigma})  \equiv H_{N}(\boldsymbol{\sigma})$.  Also in the homogeneous case, i.e. $k_a=1$, $\forall a\in\{\ {1,\dots,\nu}\}$, we recover the standard Hopfield model for $N=\sum_{a=1}^\nu N_a$ neurons and patterns $\boldsymbol{\xi}^\mu=(\boldsymbol{\xi}^{\mu,1}, \ldots , \boldsymbol{\xi}^{\mu,\nu} )$, each one obtained by concatenating the $\nu$ patterns $\boldsymbol{\xi}^{\mu,a}$, i.e. $H_{N,\nu} (\boldsymbol{\sigma}|\boldsymbol{k}=\boldsymbol{1})  \equiv H_{N}(\boldsymbol{\sigma})$.

Again, we can introduce an external field $h_a^{\mu}$, $\forall a=1, ..., \nu$ and $\forall \mu=1, ..., P$, which forces the neurons of the $a$-th group $\boldsymbol{\sigma}^a$ to retrieve of the $\mu$-th pattern $\boldsymbol{\xi}^{\mu,a}$,  as
\begin{equation}
\label{bubu3}
H_{N,\nu}(\boldsymbol{\sigma}|\boldsymbol{\alpha},\boldsymbol{k},\mathbf{h},\boldsymbol{\xi})=-\frac{N}{2}\left(\displaystyle\sum_{\mu=1}^P\sum_{a=1}^{\nu}k_a \alpha_a^2({m_a^\mu})^2 +\displaystyle\sum_{\mu=1}^P \sum_{a, b=1}^{\nu} \alpha_a\alpha_b m_a^\mu m_b^\mu\right) - N\sum_{\mu=1}^P   \sum_{a=1}^{\nu} \alpha_a m_a^{\mu}h_a^{\mu}.
\end{equation}  
Again, setting $\nu=1$ or $\boldsymbol{k}=\boldsymbol{1}$ (and homogeneous fields $h^\mu_a=h^\mu$ for each species), we recover (\ref{HopfieldHh}). As before we shall drop the dependence on the system parameters whenever unnecessary.\\
In (\ref{bubu2}) we exploited the definition (\ref{eq:ma}) to write $H_{N,\nu}(\boldsymbol{\sigma})$ in terms of the Mattis magnetizations: notice that the free field Hamiltonian is a quadratic form of $\mathbf{m}$ and can be written as 
\begin{equation} \label{eq:quadratic}
H_{N,\nu}(\boldsymbol{\sigma}) = - \frac{N}{2} \sum_{\mu=1}^P  (\mathbf{m}^{\mu},  \mathbf{J} \mathbf{m}^{\mu})= - \frac{N}{2} \mathbf{m}^{\textrm{T}} \; \mathbf{J} \; \mathbf{m},
\end{equation}
where $\mathbf{m}^{\mu} = (m^{\mu}_1, ..., m^{\mu}_{\nu})$, $\mathbf{m}$ is the $\nu \times P$ matrix whose columns are $\mathbf{m}^{\mu} $  and $\mathbf{J}$ is the interaction matrix with entries
\begin{equation}
\textbf{J}_{ab}=\begin{cases}
                & k_a \alpha_a^2 \; \; \; \ \ \   \text{if } \; \; a=b
                \\
                & \alpha_a\alpha_b \; \; \; \;\;\;\;  \text{if } \;\; a\neq b
                .
                \end{cases}
                \label{gei}
\end{equation}
Analogously to the standard Hopfield model, we define the partition function and the pressure of the model as 
\be\label{freedef}
Z_{N,\nu}(\b,\boldsymbol{\alpha},\boldsymbol{k}, \boldsymbol{h}| \boldsymbol{\xi})= \sum_{\boldsymbol{\s}}\exp\{-\b H_{N,\nu}(\boldsymbol{\sigma})\};\ \ \ \ 
f_{N,\nu}(\b,\boldsymbol{\alpha},\boldsymbol{k}, \boldsymbol{h}| \boldsymbol{\xi})=\frac 1 N \ln Z_{N,\nu}(\b,\boldsymbol{\alpha},\boldsymbol{k}, \boldsymbol{h}| \boldsymbol{\xi}).
\ee
The interaction matrix (\ref{gei}) has, in general, not definite sign for small values of $k_a$, meaning higher disassortativity in the interactions among groups. This lack of convexity makes the analysis of the free energy highly non-trivial, since the rigorous techniques typically adopted strongly rely on the control of the order parameters fluctuations, the latter being an easy task using convexity properties (see e.g., \cite{barra0,guerra-HJ}).  
The main result of this paper is the following 
\begin{thm}\label{THM}
Let $c>1+ \displaystyle\frac{\nu-1}{k^*}$, where $k^{*}=\min \{ k_a\}_{a=1}^\nu$ and $\textbf{J}^c= c \ \textbf{diag}(k_1\alpha_1^2, .. . . , k_\nu\alpha_\nu^2) -\textbf{J}$. The thermodynamic limit of the pressure of the multi-species Hopfield model described by the Hamiltonian (\ref{bubu3}) is given a.s.  by
\begin{eqnarray}\label{thmfree}
 f_{\nu}(\beta,\boldsymbol{\alpha},\boldsymbol{k}, \boldsymbol{h})&=&\lim_{N\to\infty}\frac 1 N \ln Z_{N,\nu}(\b,\boldsymbol{\alpha},\boldsymbol{k}, \boldsymbol{h}| \boldsymbol{\xi})\nonumber \\
&=& \displaystyle\sup_{{\textbf{m}}\in \mathbb{R}^{\nu\times P}} \biggl( \b\frac{\textbf{m}^T \textbf{J}^c\textbf{m}}{2} 
  +\displaystyle\sum_{a=1}^{\nu}\alpha_a f(\beta  ck_a\alpha_a \ , \beta \boldsymbol{h}_a \ - (\beta(\textbf{A}^{-1}\textbf{J}^c{\textbf{m}})_a)\biggr).
\label{ottaviuccia}
\end {eqnarray}
where  $f(t,\textbf{y})$ represents the pressure of the standard Hopfield model at a given temperature $t$ and with a suitable external field $\textbf{y}$, and $\mathbf{A} = \textbf{diag}(\alpha_1,...,\alpha_{\nu})$. The optimal order parameters $\textbf{m}$,  are the solution of the self-consistent equations 
\begin{equation}
\textbf{m}_a=\biggl<\boldsymbol{\xi} \tanh\biggl[\beta\boldsymbol{\xi}\cdot(\textbf{A}^{-1}\textbf{J}{\textbf{m}}+\boldsymbol{h})_a  \biggr]\biggr>_{\boldsymbol{\xi}\in\{-1,1\}^{P}}.
\label{fiorellona}
\end{equation}
\end{thm}

As a consistency check, by setting $k_a = 1$, $\boldsymbol{h}_a=\boldsymbol{h}$ and posing $M^{\mu} = \sum_{a=1}^{\nu} m_a^{\mu} \alpha_a$ for any $a=1,...,\nu$ and $\mu=1,...,P$, we get
\begin{eqnarray}
\boldsymbol{M}  = \sum_{a=1}^{\nu}  \alpha_a \left \langle \boldsymbol{\xi} \tanh \left [\beta \boldsymbol{\xi} \cdot (\boldsymbol{M}+\boldsymbol{h}) \right]\right \rangle_{\boldsymbol{\xi}}= \left \langle \boldsymbol{\xi} \tanh \left [\beta \boldsymbol{\xi} \cdot (\boldsymbol{M}+\boldsymbol{h}) \right]\right \rangle_{\boldsymbol{\xi}},
\label{fiorellona2}
\end{eqnarray}
which are the standard Hopfield magnetizations (\ref{gaiaa}) as expected. In general,  plugging  the solution of ($\ref{fiorellona}$) into  ($\ref{thmfree}$)  we get
\be\label{freefin}
f_{\nu}(\beta,\boldsymbol{\alpha},\boldsymbol{k}, \boldsymbol{h})=\ln 2
  +\displaystyle\sum_{a=1}^{\nu}\alpha_a \meanv{\log\cosh( \beta \boldsymbol{\xi}\cdot(\textbf{A}^{-1}\textbf{J}{\textbf{m}} +\boldsymbol{h})_a)}_{\boldsymbol{\xi}}
-\b\frac{\textbf{m}^T \textbf{J}\textbf{m}}{2}. 
\ee
Taking the derivative  of $(\ref{freefin})$ with respect to the external fields, one can see that the solution of ($\ref{fiorellona}$) has the meaning of the averaged Mattis magnetization in the thermodynamic limit. It is important to note that ($\ref{freefin}$) does not depend on $c$ as it should.  In fact, $c$ is just needed to get the correct convexity of the variational principle ($\ref{thmfree}$), i.e. it folds the trial free energy  leaving its extremal point unchanged.

In the last section we will present the proof of Theorem \ref{THM} based on the introduction of a suitable perturbation, induced by $c$,  in such a way that the resulting Hamiltonian becomes a convex form of the $\mathbf{m}$'s.  This way we will be able to map the problem into a mechanical framework in such a way that the free energy can be obtained as a solution of a Hamilton-Jacobi equation. Before proving the Theorem (see Sec.~\ref{sec:solution}), in the next Section we present some special cases of the Multi-species Hopfield mode.

\section{Special cases}\label{sec:special}
In this section we show two particular cases for the multi-species Hopfield model (\ref{bubu}), namely the  bi-directional associative memory (BAM) model and the three-layer Restricted Boltzmann Machine (RBM), two interesting generalization of the Hopfield model introduced  in the context of, respectively, biological neural networks and artificial neural networks and Deep Learning.

\subsection{The BAM model }
As we mentioned in Sec.~\ref{sec:Hopfield}, the Hopfield model reproduces the ability of the human brain to recall information from corrupted or partial data, yet human memory has many other characteristics. Among these is the ability to recall information coupled to one another: for example, if we remember the surname of someone we know, listening to that surname, even in the presence of some errors, we will probably be able to recall both the surname and the name of this person. The Hopfield model does not show this kind of behavior. The  bi-directional associative memory (BAM) model, introduced by Kosko \cite{kosko1988bidirectional} in 1988, is the simplest network with this property: it is a neural network capable of storing and associating pairs of data, encoded by pairs of binary strings, of length $N$ and $M$ respectively, using a network of $N+M$ neurons, as we are going to review briefly. 
\newline
In this model one distinguishes two different types of neurons: the \emph{input neurons} $\sigma_i $, labelled by $i \in \{1, ..., N\}$, and the \emph{output neurons} $\sigma_i$, labelled by $i \in \{1+N,. . . , M+N\}$. The connections between the input and output neurons depend on the $P$ pairs of patterns $(\boldsymbol{\xi}, \boldsymbol{\eta})$ with $N$ and $M$ components, respectively. Each component $\xi_i^{\mu}$, $i=1, ..., N, \mu=1,..., P$ and $\eta_i^{\mu}$, $i=N+1, ..., N+M, \mu=1,..., P$ are independent random variables and identically distributed with zero average.
The synaptic matrix of the interactions between the spin $i \in \{1, ..., N\}$ and the spin $j \in \{N + 1, ..., N + M\}$ is given by 
\begin{equation}\label{eq:BAMJ}
J_{ij} =\begin{cases}
                &\frac{1}{\sqrt{MN}} \sum_{\mu=1}^P \xi_i^{\mu} \eta_j^{\mu} ~~~ \textrm{if} ~~ i \in \{1, ..., N \} ~~ \textrm{and} ~~ j \in \{N+1, ..., N+M \}.
                \\
                & 0 ~~~ \text{otherwise}
                .
                \end{cases}
\end{equation}
Each set of neuron can be associated to an order parameter, referred to as $\mathbf{m}$ and $\mathbf{n}$, respectively, whose $\mu$-th entry quantifies the overlap between the spin configuration and the related patterns, that is
\begin{equation}
m_{\mu} = \frac{1}{N} \sum_{i=1}^N \sigma_i \xi_i^{\mu}, ~~~~~  n_{\mu} = \frac{1}{M} \sum_{i=N+1}^{N+M} \sigma_i \eta_i^{\mu}
\end{equation}
The Hamiltonian describing this system reads as
\begin{equation} \label{HBAM}
H_{N,M}(\boldsymbol{\sigma}|\boldsymbol{\xi},\boldsymbol{\eta}) =\frac{1}{\sqrt{NM}}\sum_{\mu=1}^P \sum_{i=1}^N\sum_{j=1}^M  \xi_i^\mu\eta_j^\mu \sigma_i\sigma_j=\sqrt{MN} \sum_{\mu=1}^P m_\mu n_\mu,
\end{equation}
and the dynamics (\ref{eq:dinamica}) can be adopted straightforwardly to make the system relax toward an equilibrium configuration described by the related Boltzmann-Gibbs distribution.
The recall phase occurs analogously as for the Hopfield model: if, given a certain initial point, the spin configuration converges to a state such that $\sigma_i=\xi_i^{\mu}$ $\forall i=1,...,N$ and $\sigma_i=\eta_i^{\mu}$ $\forall i=N+1,...,N+M$, the information, codified by the strings $\boldsymbol{\xi}^{\mu}$ and $\boldsymbol{\eta}^{\mu}$, is retrieved. The BAM model was investigated extensively by \cite{kurchan1994statistical,englisch1995bam} addressing the low-storage (i.e., $P/M$ and $P/N$ both vanishing in the thermodynamic limit) as well as the high-load (i.e., $P/M$ and $P/N$ both finite in the thermodynamic limit) regimes. In particular, in the former case they found that the order parameters fulfil the following self-consistent equations
\begin{align}
\label{iphone}
&\textbf{m}= \langle \boldsymbol{\xi} \tanh(\beta \textstyle\sqrt{\frac{M}{N}} \ \boldsymbol{\xi}\cdot\textbf{n}) \rangle_{\boldsymbol{\xi}},
\\
\label{iphonex}
&\textbf{n}= \langle \boldsymbol{\eta} \tanh(\beta \textstyle\sqrt{\frac{N}{M}} \ \boldsymbol{\eta}\cdot\textbf{m}) \rangle_{\boldsymbol{\eta}},
\end{align}
and, in the latter case, they derive an expression for the critical capacity, namely the maximum number of patterns that the network is able to safely handle: this quantity is shown to depend on the ratio $M/N$ and it is maximum for $M=N$\footnote{The critical threshold for the BAM model occurs to be lower than that for the standard Hopfield model and this result may be explained by noticing that storing a pair of patterns corresponds to doubling the information stored per pattern and therefore the total information stored in the $(N+M)^2/2$ synaptic connections.}. Further results can be found in \cite{Liao,Cao,Cao2,Cao3,Cao4,Park}.

%
The Hamiltonian (\ref{HBAM}) can be recast in the general description provided by (\ref{bubu}), by setting $\nu=2$ and $k_1 = k_2 =0$. Otherwise stated, we are considering two groups where interactions only occur between neurons belonging to different groups. Moreover, 
$\boldsymbol{\xi}^{\mu} = \boldsymbol{\xi}^{\mu,1}$ and $\boldsymbol{\eta}^{\mu} = \boldsymbol{\xi}^{\mu,2}$; $\boldsymbol{\sigma}^1 = \{ \sigma_1, ..., \sigma_N \}$ and 
$\boldsymbol{\sigma}^2 = \{ \sigma_{N+1}, ..., \sigma_{N+M} \}$; $\alpha_1=\frac{N}{N+M}$ and $\alpha_2=\frac{M}{N+M}$.
By setting these parameters in the self-consistent equations (\ref{fiorellona}), simply rescaling the temperature as $\b\to\b/{\sqrt{\a_1\a_2}}$ because of the different normalizations in ($\ref{HBAM}$) and ($\ref{bubu2}$)  we exactly recover (\ref{iphone}) and (\ref{iphonex}).

 \begin{figure}[tb]
\begin{center}
\includegraphics[scale=0.4]{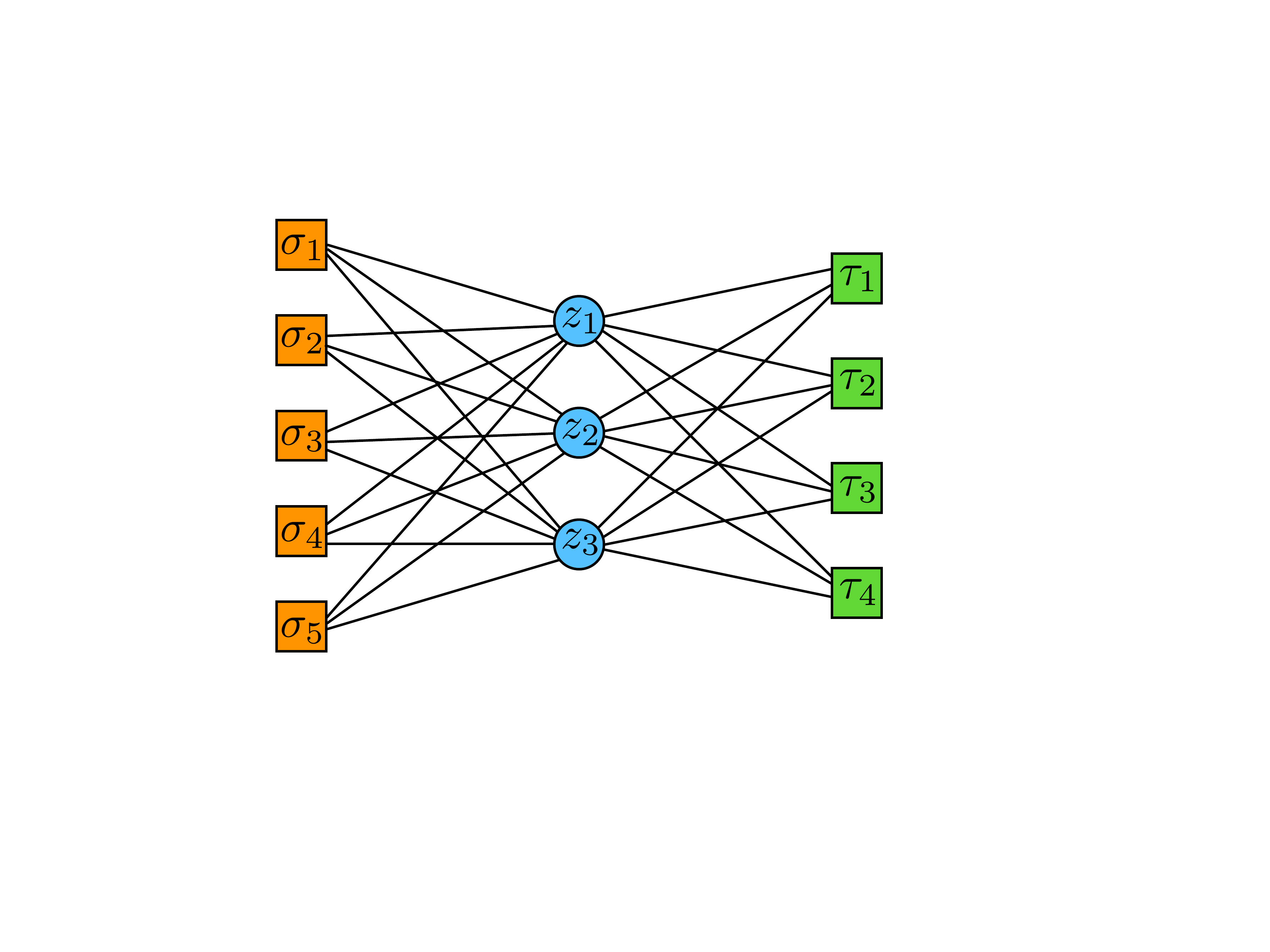}
\caption{Sketch of a RBM with three layers. The visible layer (on the left) is made of $N=5$ spins ($\sigma_i$, $i=1,...,N$), the output layer (on the right) is made of $M=4$ spins ($\sigma_i$, $i=1,...,N$), and the hidden layer (in the middle) is made of $P=3$ spins ($z_i$, $i=1,...,P$). In our analysis we assumed boolean visible and output  neurons ($\boldsymbol{\sigma} \in \{-1,+1\}^N$ and $\boldsymbol{\tau} \in \{-1,+1\}^M$) while Gaussian neurons ($\boldsymbol{z} \in \mathbb{R}^P$) in the hidden layer.}
\label{fig:rbm}
\end{center}
\end{figure}

\subsection{Three-layer Restricted Boltzmann Machines} \label{sec:special2}
In Sec.~\ref{sec:Hopfield} we presented the Hopfield network as a standard model for associative memory, while the basic model for neural networks learning processes is the Boltzmann machine (BM). Remarkably, retrieval and learning represent two complementary aspects of the cognitive process. In particular, the Hebb rule (\ref{hebbrule}) can be shown to emerge from a learning process accomplished through a BM \cite{BBCS,ABDG}. The recent success of Deep Learning in a wide range of  different applicative sciences poses the attention on a specific class of BMs, the Deep Boltzmann Machines \cite{Bengio}. They are composed of multiple layers of units (neurons), where units belonging to any two consecutive layers are connected. The elementary component of this architecture is thus a particular two layers BM, also referred to as Restricted Boltzmann Machine (RBM).
 
More precisely, let us consider a group of $N$ binary units $\si \in \{-1, +1\}$, $i \in \{1, ..., N\}$, which represent the neurons of the first layer (also called visible), a group of units $z_{\mu}$, $\mu \in \{1,. . . , P\}$, which represent the neurons of the second layer (also called hidden), and a set of weights $w_i^{\mu}$ associated to the links connecting the neurons $\s_i$ and  $z_\mu$. 
For the scope of this section we consider the hidden units to be real variables, distributed with a Gaussian prior. All the results can be suitably generalized to the case of different priors \cite{BGST1, BGST2}.
We can then define the Hamiltonian of the restricted Boltzmann machine as 
\begin{equation}
H_N(\boldsymbol{\sigma}, \textbf{z}| \textbf{w})= -\frac{1}{\sqrt{N}} \displaystyle\sum_{i=1}^N\sum_{\mu=1}^P w_i^\mu \sigma_i z_\mu.
\end{equation} 
The associated Gibbs distribution  is used typically to fit some data distribution over the visible layer, optimizing over the weights. Thus a RBM can be used as a parametric probabilistic model in an unsupervised learning framework \cite{gabrie1}.  For a given realization of the weights  the related intensive pressure of a RBM reads as
\begin{equation}
f_N(\beta|\textbf{w}) =\frac{1}{N} \ln \displaystyle\sum_{\boldsymbol{\sigma}}\prod_{\mu} \int_{-\infty}^{\infty} dz_\mu\frac{1}{\sqrt{2\pi}}\exp\biggl(-\frac{z_\mu^2}{2}\biggr) \exp\biggl(  \frac{\beta}{\sqrt{N}} \sum_\mu\sum_i w_i^\mu \sigma_i z_\mu\biggr).
\end{equation}
Since there are no weights between neurons belonging to the same layer, we can marginalize over the $z$ variables by evaluating the Gaussian integral obtaining
\begin{align}
\nonumber
f_N(\beta|\textbf{w}) = \frac{1}{N} \ln \displaystyle \sum_{\boldsymbol{\sigma}} \exp \Biggl\{\frac{\beta^2}{2N}\sum_{i,j=1}^N \biggl( \sum_{\mu=1}^P w_i^\mu w_j^\mu\biggr)\sigma_i\sigma_j   \Biggr\}.
\end{align}
Comparing this result with the pressure of the standard Hopfield model (\ref{HopfieldH}), we can see that the two expressions are equivalent (apart for the value of the temperature) identifying the weights of the RBM $\mathbf{w}^{\mu}$ with the patterns $\boldsymbol{\xi}^{\mu}$ of the Hopfield network. This mapping was used in \cite{BGST1,BGST2} for studying the performances of RBM's in learning processes using the knowledge of the Hopfield model phase diagram. Hereafter, we use the same argument to show that a 3-layer RBM is equivalent to a 2-species Hopfield model. This is in agreement with a recent finding by M.Mezard \cite{mezard} that has shown, on the other hand, how a 2-layers RBM with combinatorial weights can be seen as a 3-layers RBM.

Let us consider a RBM made of three layers: let $\sigma_i$, $i \in \{\ 1,\dots, N \}$, the neurons belonging to the first (visible) layer, $z_\mu$, $\mu\in\{\ 1,\dots, P \}$, the neurons belonging to the hidden layer and $\tau_k$, $k\in\{\ 1,\dots, M \}$, the neurons belonging to the third (output) layer, as sketched in Fig.~\ref{fig:rbm}.
We focus on the regime $P \ll N,M$ and we take $\sigma_i,\tau_k\in \{-1, +1 \}$ and $z_\mu\in\mathbb{R}$ as Gaussian variables for simplicity.  This kind of structure is often referred to as an autoencoder.
The Boltzmann-Gibbs distribution associated to this system is 
\begin{equation} 
p_{M,N}(\boldsymbol{\sigma}, \boldsymbol{\tau}, \textbf{z}|\boldsymbol{\xi}, \boldsymbol{\eta})= \frac {1}{Z_{M,N}}\prod_{\mu}\frac{1}{\sqrt{2\pi}}\exp\biggl(-\frac{z_\mu^2}{2}\Biggr)\exp\Biggl\{{{\frac {\sqrt{\beta}}{\sqrt{N}}\displaystyle\sum_{i=1}^N\sum_{\mu=1}^P{{\xi^\mu_i}\sigma_iz_\mu}}+ {\frac {\sqrt{\b}}{\sqrt{M}}\displaystyle\sum_{k=1}^M\sum_{\mu=1}^P {{\eta^\mu_k}\tau_k z_\mu}}}\Biggr\}.
\label{dist prob}
\end{equation}
Since neurons belonging to the same layer are not interacting each other, we can marginalize the joint probability (\ref{dist prob}) over $z$ exploiting the Gaussian integration as
\begin{eqnarray}
\label{ciao}
p_{M,N}(\boldsymbol{\sigma}, \boldsymbol{\tau}|\boldsymbol{\xi}, \boldsymbol{\eta}) &=& \frac{1}{Z_{M,N}}  \exp \Biggl\{{\frac{\b}{2N} \displaystyle\sum_{\mu=1}^P\sum_ {i,j=1}^{N,N} {\xi_i}^\mu{\xi_j}^\mu \sigma_i \sigma_j }+  {\frac{\b}{2M} \displaystyle\sum_{\mu=1}^P\sum_ {k,l=1}^{M,M} {\eta_k}^\mu{\eta_l}^\mu \tau_k \tau_l} 
\nonumber \\
&+& \frac{\b}{\sqrt{NM}}\displaystyle \sum_{\mu=1}^P\sum_{i,k=1}^{N,M}{\xi_i}^\mu{\eta_k}^\mu \sigma_i \tau_k\Biggr\} .
\end{eqnarray}  
\par\noindent 
Comparing (\ref{ciao}) with (\ref{eq:peq}) one can see that the 3-layers RBM considered here is equivalent to a 2-species Hopfiled model with Hamiltonian
\begin{align} 
\label{hamilton}
H_{M,N} (\boldsymbol{\sigma},\boldsymbol{\tau})&= -\biggl(\frac{1}{2N} \displaystyle\sum_{\mu=1}^P\sum_ {i,j=1}^{N,N} {\xi_i}^\mu{\xi_j}^\mu \sigma_i \sigma_j +  \frac{1}{2M} \displaystyle\sum_{\mu=1}^P \sum_ {k,l=1}^{M,M} {\eta_k}^\mu{\eta_l}^\mu \tau_k \tau_l + \frac{1}{\sqrt{NM}} \displaystyle\sum_{\mu=1}^P \sum_{i,k=1}^{N,M}{\xi_i}^\mu{\eta_k}^\mu \sigma_i \tau_k\biggr).
\end{align}
\noindent Introducing the two sets of order parameters 
\begin{equation}
m_\mu=\frac{1}{N}\displaystyle\sum_{i=1}^N {\xi_i}^\mu \sigma_i, \ \ \ \ n_\mu=\frac{1}{M}\displaystyle\sum_{k=1}^M {\eta_k}^\mu \tau_k.
\end{equation}
the expression (\ref{hamilton}) can be recast as 
\begin{equation} \label{HH}
H_{M,N}(\textbf{m},\textbf{n})= -{\frac{N}{2} \sum_{\mu=1}^P{m_\mu}^2}-{\frac{M}{2}\sum_{\mu=1}^P{n_\mu}^2} -{\sqrt{NM}\sum_{\mu=1}^Pm_\mu n_\mu},
\end{equation}
\par\medskip\noindent 
which corresponds to a multi-species Hopfield model  (\ref{bubu}) as long as we set $\nu=2$, $\b\to\b(M+N)/\sqrt{MN}$, $k_1=\sqrt{M/N}$ and $k_2=\sqrt{N/M}$. The solution of this system can therefore be derived from (\ref{ottaviuccia}) by properly setting the parameters.
To fix ideas we assume, without loss of generality, that $M=\gamma N$, where $\gamma\in\mathbb{R}^+$. 
Thus, recalling (\ref{fiorellona}), we find
\begin{equation} \label{m1}
\mathbf{m}=\left \langle\boldsymbol{\xi} \tanh \left[ \b \sum_{\mu=1}^P(m_{\mu}+\sqrt{\gamma} n_{\mu}) \cdot\boldsymbol{\xi} \right] \right \rangle_{\boldsymbol{\xi}},
\end{equation}
\begin{equation}  \label{n1}
\mathbf{n}=\left \langle\boldsymbol{\eta} \tanh \left[ \frac{\b}{\sqrt{\gamma}} \sum_{\mu=1}^P \left(m_{\mu}+\sqrt{\gamma} \ n_{\mu} \right) \cdot\boldsymbol{\eta} \right] \right \rangle_{\boldsymbol{\eta}}.
\end{equation}


The structure of (\ref{m1}) and (\ref{n1}) suggests that we can linearly combine $\mathbf{m}$ and $\mathbf{n}$ to get a unique order parameter, referred to as $\mathbf{p}$ and whose $\mu$-th component reads as $p_{\mu} = m_{\mu} + \sqrt{\gamma} n_{\mu}$, $\mu=1,...,P$. As a result, (\ref{m1}) and (\ref{n1}) can be recast as
\begin{equation}\label{gaia}
\mathbf{p} = \left \langle \boldsymbol{\xi} \tanh (\b \boldsymbol{\xi}\cdot \textbf{p})\right \rangle_{\boldsymbol{\xi}} +\sqrt{\gamma} \ \left \langle \boldsymbol{\eta}\tanh \biggl(\frac{\b}{\sqrt{\gamma}} \boldsymbol{\eta}\cdot\textbf{p}\biggr)\right \rangle_{\boldsymbol{\eta}}
\end{equation}
This self-consistent equation was solved numerical and results are shown in Fig..~\ref{fig:P1}.
\begin{figure}[tb]
\begin{center}
\includegraphics[scale=0.4]{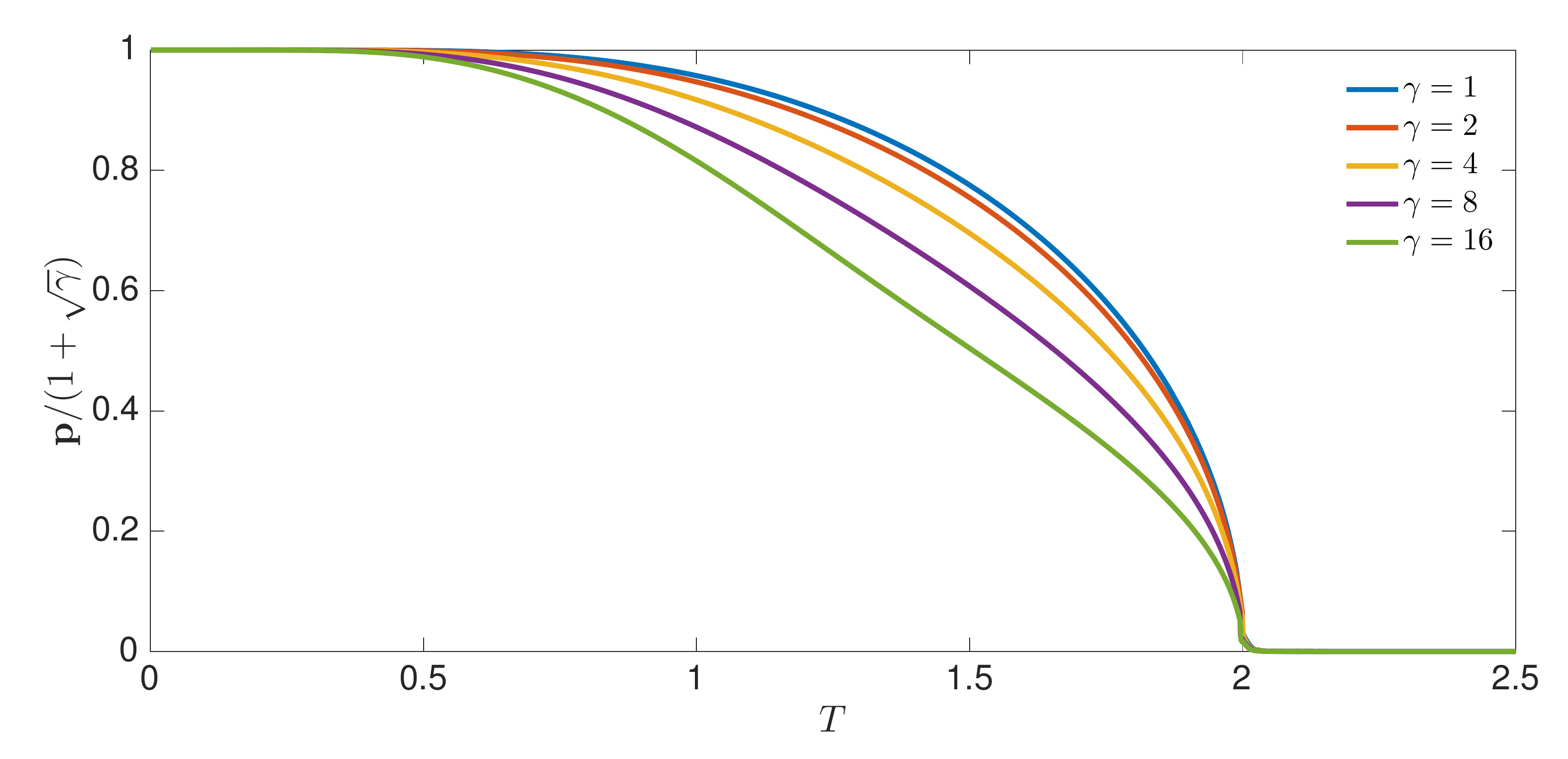}
\caption{This plot shows the behavior of the normalized overlap $p_{\mu}/(1 + \sqrt{\gamma})$ related to the retrieved pattern, as a function of the temperature $T$, obtained by solving the self-consistent equation (\ref{gaia}) for $P=2$. In this case one pattern, say $\mu=1$ to fix ideas, is retrieved and the related overlap follows the outline shown in this figure; the remaining pattern $\mu=2$ is not retrieved and the related overlap  is null for any $T$. Several values of the parameter $\gamma$ are considered and shown in different colors, as explained by the legend.} 
\label{fig:P1}
\end{center}
\end{figure}
These plots suggest a phase transition at $T_c=2$, such that, as $T>T_c$, the magnetization vanishes and retrieval is no longer achievable. In fact let us notice that 
\begin{align}
\textbf{p}^2&=\Big\langle \boldsymbol{\xi} \tanh (\b \boldsymbol{\xi}\cdot \textbf{p})\Big\rangle_{\boldsymbol{\xi}} \cdot\textbf{p} +\sqrt{\gamma} \ \Big\langle \boldsymbol{\eta}\tanh \biggl(\frac{\beta}{\sqrt{\gamma}} \boldsymbol{\eta}\cdot\textbf{p}\biggr)\Big\rangle_{\boldsymbol{\eta}}\cdot \textbf{p}
\\
\nonumber
& = \Big\langle \xi\cdot\textbf{p} \tanh (\beta \boldsymbol{\xi} \cdot \textbf{p})\Big\rangle_{\boldsymbol{\xi}} +\sqrt{\gamma} \ \Big\langle \boldsymbol{\eta}\cdot\textbf{p}\tanh \biggl(\frac{\beta}{\sqrt{\gamma}} \boldsymbol{\eta}\cdot\textbf{p}\biggr)\Big\rangle_{\boldsymbol{\eta}},
\end{align}
and, recalling that $\lvert \tanh(z)\rvert \le \lvert z\rvert$, we get
\begin{align}
\label{cosascrivo}
\textbf{p}^2 & \le \beta\bigl(\langle (\boldsymbol{\xi} \cdot \textbf{p})^2\rangle_{\boldsymbol{\xi}}+ \langle (\boldsymbol{\eta} \cdot \textbf{p})^2\rangle_{\boldsymbol{\eta}}\bigr)= 2\beta \langle (\lambda \cdot \textbf{p})^2\rangle_{\boldsymbol{\lambda}} =2\beta \textbf{p}^2.
\end{align}
\par\noindent  We can therefore conclude that as $\beta < \frac{1}{2}$, namely as $T>2$,  the only equilibrium solution for the problem is $\textbf{p}=\textbf{0}$.
We want to check further that the critical temperature of the model is $T_c=2$. 
Since
\begin{equation*}
\langle \xi^\mu \rangle _{\boldsymbol{\xi}} =0 \;\;\;\; \text{e} \;\;\;\; \langle \xi^\mu\xi^\nu \rangle_{\boldsymbol{\xi}} =\delta_{\mu\nu} \ \ \ \langle \xi^\mu\xi^\nu\xi^\rho\xi^\sigma\rangle_{\boldsymbol{\xi}}=\delta_{\mu\nu}\delta_{\rho\sigma} + \delta_{\mu\rho}\delta_{\nu\sigma} +\delta_{\mu\sigma}\delta_{\rho\nu}-2\delta_{\mu\rho}\delta_{\rho\nu}\delta_{\nu\sigma},
\end{equation*}
expanding (\ref{gaia}) around $p_\mu = 0$ we get 
\begin{align*}
p_\mu =
 \beta p_\mu +\beta^3 p_\mu \bigl(-\textbf{p}^2 +\frac{2}{3} p_\mu^2\bigr) + \beta p_\mu + \frac{\beta^3}{\gamma} p_\mu \bigl(-\textbf{p}^2+ \frac{2}{3} p_\mu^2\bigr).
\end{align*}
Asking for a pure state retrieval, namely $\textbf{p}^2 =p_\mu^2$, we get
\begin{equation*}
p_\mu(2\beta -1)- \beta^3  p_\mu \biggl( 1+\frac{1}{\gamma}\biggr)\frac{1}{3}p_\mu^2 =0.
\end{equation*}
Replacing $\beta= 1/T$, the non trivial solutions are
\begin{equation}
p_\mu=\pm \sqrt{\frac{3(T-2)}{1+\displaystyle\frac{1}{\gamma}}},
 \end{equation}
i.e. the critical temperature  is $T_c=2$. Moreover, one can see that $p_\mu$ scales as $\sqrt{T-2}$ and therefore, the critical exponent turns out to be $1/2$, which suggests that the phase transition is of the second order.


 \section{Solution of the multi-species Hopfield model}  \label{sec:solution}

We investigate the model introduced in the previous section focusing on the low-load regime, namely assuming $\lim_{N\rightarrow \infty} P/N =0$. As underlined in Section \ref{sec:model}, the Hamiltonian (\ref{bubu}) under study is a quadratic form with non-definite sign. However, we can ``adjust'' it by inserting a suitable, additional term that enhances the diagonal entries of the matrix $\mathbf{J}$ (\ref{gei}), in such a way that the contributions due to intra-group interactions are prevailing and the matrix gets positive definite.
Being $c$ a real, strictly positive parameter, we define
\begin{equation}
\textbf{J}^c= c \ \textbf{diag}(k_1\alpha_1^2, .. . . , k_\nu\alpha_\nu^2) -\textbf{J}.
\label{matj}
 \end{equation} 
Thus it holds
\begin{lemma} 
As soon as $c>1+\displaystyle\frac{\nu-1}{k^*}$, where $k^*=\min\bigl\{{k_a}\bigr\}_{a=1}^{\nu}$,  $\textbf{J}^c$ is positive definite.
Then, the matrix {\normalfont{$\textbf{J}^c$}} can be written as {\normalfont{$\textbf{J}^c=\textbf{P}^T\textbf{P}$}}, where {\normalfont{$\textbf{P}$}} is the matrix whose rows are given by the normalized eigenvectors of {\normalfont{$\textbf{J}^c$}}.
\label{lemma1}
\end{lemma}
\par\smallskip\noindent
The proof of this lemma can be found in Appendix A.
\par\smallskip\noindent

\begin{proof}{of Theorem \ref{THM}}.
\par\medskip\noindent Let $c>1+ \displaystyle\frac{\nu-1}{k^*}$ and let us transform the order parameters $\textbf{m}$ according to the  rotation defined by $\textbf{P}$  as in Lemma (\ref{lemma1}), namely $\textbf{m}^\prime =\textbf{P}\textbf{m}$, or, more explicitly, ${m_a^{\mu}}^\prime = \sum_{k=1}^{\nu} P_{ak} m_k^{\mu}$.
Following \cite{genovese2016non, guerra-HJ, HJ-Jstat,barra2013mean, genovese2009mechanical} let us also introduce the interpolating function $\Phi_{N, \nu}(t,\textbf{x})$, where the variables $t \in \mathbb{R}^+$ and $\textbf{x} \in \mathbb{R}^{\nu\times P}$ are meant, respectively, as generalized time and space:
\begin{align}
\nonumber
\Phi_{N, \nu}(t,\textbf{x}) & = \frac{1}{N}\ln\displaystyle\sum_{\boldsymbol{\sigma}}\exp \Bigg\{ \frac {Nt} 2\biggl(\displaystyle\sum_{\mu=1}^P\sum_{a=1}^{\nu}k_a\alpha_a^2 (m_a^\mu)^2+ \displaystyle\sum_{\mu=1}^P\sum_{a\neq b}^{\nu,\nu} \alpha_a\alpha_b m_a^\mu m_b^\mu  \biggl) 
\\ \label{morda2}
& \ \ \ \ 
+N\frac{( \beta-t)c}{2}\displaystyle\sum_{\mu=1}^P\sum_{a=1}^{\nu} k_a\alpha_a^2 m_a^2+ N \displaystyle\sum_{\mu=1}^P\sum_{a=1}^{\nu} x_a^\mu {m_a^\mu}^\prime\Bigg\}.
\end{align}
In order to lighten the notation we also introduce the interpolating \textit{Boltzmann factor} $B_{N,\nu}(t,\textbf{x})$ such that $\Phi_{N, \nu}(t,\textbf{x})=\frac{1}{N}\ln \sum_{\boldsymbol{\s}}B_{N,\nu}(t,\textbf{x})$ and 
the average performed with respect to $B_{N,\nu}(t,\textbf{x})$ shall be denoted as $\langle \cdot \rangle_{(t,\mathbf{x})}$.
The interpolating partition function $Z_{N,\nu}(t,\mathbf{x})$ is defined analogously.
Notice that, when $t=\beta$ and $\textbf{x}(\mathbf{h})= \beta  (\mathbf{P}^T)^{-1}  \mathbf{A} \mathbf{h}$, we recover the original model described by (\ref{bubu3}), i.e. 
\begin{equation}
f_{\nu}(\beta, \boldsymbol{\alpha}, \boldsymbol{k}, \mathbf{h}) = \lim_{N \rightarrow \infty} \Phi_{N,\nu}(\beta,\beta (\mathbf{P}^T)^{-1}\mathbf{A} \mathbf{h}), 
\end{equation}
and the average $\langle \cdot \rangle_{(t,\mathbf{x})}$ recovers the expectation over the standard Boltzmann-Gibbs measure.
\par\noindent 
We now show that $\Phi_{N,\nu}(t,\textbf{x})$ fulfills an Hamilton-Jacobi equation; to this task we first evaluate the derivate of (\ref{morda2}) with respect to $t$
\begin{align}
\partial_{t}\Phi_{N, \nu}(t,\textbf{x}) &= \frac{1}{2 Z_{N,\nu}(t,\mathbf{x})}\displaystyle\sum_{\boldsymbol{\sigma}} \biggl[ (1-c)\displaystyle\sum_{\mu=1}^P \sum_{a=1}^{\nu} k_a \alpha_a^2{m_a^\mu}^2  +\displaystyle\sum_{\mu=1}^P \sum_{a\neq b}^{\nu,\nu} \alpha_a\alpha_b m_a^\mu m_b^\mu\biggr]  B_{N,\nu}(t,\textbf{x})
\\ 
\label{poten0}
&= -\frac{1}{2}\langle\textbf{m},\textbf{J}^c\textbf{m}\rangle  =- \frac{1}{2}\langle\textbf{P}\textbf{m},\textbf{P}\textbf{m}\rangle
= -\frac{1}{2}\displaystyle\sum_{a=1}^{\nu}\sum_{\mu=1}^P \langle({m_a^{\mu}}^{\prime})^2\rangle,
\end{align}
and the derivative with respect to $x_{a}^{\mu}$ , $a=1, \dots , \nu$: 
\begin{equation}
\partial_{x_a^{\mu}}\Phi_{N,\nu}(t,\mathbf{x}) =\frac{1}{Z_{N,\nu}(t,\mathbf{x})}\displaystyle\sum_{\boldsymbol{\sigma}} m_a^{\mu\prime} B_{N,\nu}(t,\textbf{x})=  \displaystyle\langle{m_a^{\mu}}^\prime\rangle .
\label{der}
\end{equation}
Combining (\ref{poten0})-(\ref{poten}), we see that the following Hamilton-Jacobi equation for $\Phi_{N, \nu}(t,\textbf{x})$ holds by construction
\begin{equation}
 \partial_{t}\Phi_{N,\nu}(t,\textbf{x}) +\frac{1}{2}{\mid\nabla\Phi_{N, \nu}(t,\textbf{x})\mid}^2 + V_N(t,\textbf{x}) =0 .
\label{Hjag}
\end{equation}
where we have introduced the potential $V_N(t,\textbf{x})$ as the sum of fluctuations of the rotated magnetizations with respect to the average $\langle \ \cdot \ \rangle_{(t,\mathbf{x})}$ as
\begin{equation}
V_N(t,\textbf{x}) = \frac{1}{2N}\displaystyle\Delta \Phi_{N,\nu}(t,\textbf{x}) = \frac{1}{2}\displaystyle\sum_{a=1}^{\nu}\sum_{\mu=1}^P [\langle({{m_a^{\mu}}^\prime})^2\rangle - {\langle{m_a^{\mu}}^\prime\rangle}^2 ].
\label{poten}
\end{equation}
Otherwise stated, $\Phi_{N, \nu}(t,\textbf{x})$ plays as the action potential for a particle moving in $\mathbb{R}^{\nu \times P}$ in the presence of a potential $V_N(t,\textbf{x})$. In this mechanical analogy, $\partial_{x_a^{\mu}}\Phi_{N, \nu}(x,t)$ plays as the velocity of the particle.
\par\noindent As $N\to\infty$, the solution of this PDE can be shown to approach the unique viscosity solution
of the free Hamilton-Jacobi equation, given by the Hopf-Lax formula (see for instance 
\cite{evans1998partial,lax})
\begin{equation}
\Phi_{\nu}(t,\textbf{x}): = \lim_{N \rightarrow \infty} \Phi_{N,\nu}(t,\textbf{x}) =\displaystyle\min_{\textbf{y}\in \mathbb{R}^{\nu\times P}} \biggl( -\frac{(\textbf{x}-\textbf{y})^2}{2t} +\Phi_{N,\nu}(0,\textbf{y})\biggr).
\end{equation}
\par\noindent Now, since as $N \rightarrow \infty$ the motion is free, we can relate the position $\mathbf{y}$ at time $t=0$ to the position $\mathbf{x}$ at the arbitrary time $t$ by $\textbf{y}=\textbf{x}-t \ \textbf{m}^\prime$, that is  
\begin{equation}
\Phi_{\nu}(t,\textbf{x})=\displaystyle\min_{{\textbf{m}^\prime}\in \mathbb{R}^{\nu\times P}} \biggl( -\frac{t(\textbf{m}^\prime)^2}{2} +\Phi_{N,\nu}(0,\textbf{x} - t \ \textbf{m}^\prime)\biggr).
\label{variaz0}
\end{equation}
\par\medskip\noindent
We can calculate explicitly $\Phi_{N,\nu}(0,\textbf{y})$ from (\ref{morda2}) as
\begin{align}
\nonumber
\Phi_{N,\nu}(0,\mathbf{y})& = \frac{1}{N}\ln\displaystyle\sum_{\boldsymbol{\sigma}}\exp \Bigg\{ \frac{N}{2} \displaystyle\sum_{a=1}^{\nu} \sum_{\mu=1}^P \  c\beta k_a \alpha_a^2 \  (m_a^{\mu})^2 +N \displaystyle\sum_{a=1}^{\nu}\sum_{\mu=1}^P y^\mu_a (\textbf{P}\textbf{m}^\mu )_a \Bigg\}\\
\label{combi}
&=\displaystyle\sum_{a=1}^{\nu}\frac{ \alpha_{a} }{N_{a}}\ln\sum_{\boldsymbol{\sigma}_{a}} \exp \left\{ \frac{N_a}{2} \displaystyle\sum_{\mu=1}^P \  c\beta k_a\alpha_a \  (m_a^{\mu})^2 + N_a\displaystyle\sum_{\mu=1}^P  (\textbf{A}^{-1}\textbf{P}^T \boldsymbol{y}^\mu)_a  m^\mu_a\right\},
\end{align}
%
%
where in the last line 
we highlighted the contribution stemming from each group. In fact, (\ref{combi}) turns out to be the linear combination of the  pressures of standard Hopfield models of different sizes,  at different temperatures and different external fields. This can be seen recalling (\ref{PF}) and (\ref{effe}). 
In particular, in this configuration, the $a$-th term corresponds to a Hopfield model with inverse temperature $\beta  c k_a\alpha_a$ and external field $((\textbf{A})^{-1}\textbf{P}^T\textbf{y})_a$ 
Thus, in the thermodynamic limit, it holds
 \begin{align} \label{dato0}
 \Phi_{\nu}(0,\textbf{x} - t \textbf{m}^\prime)&=\displaystyle\sum_{a=1}^\nu\alpha_a f(  c \beta k_a \alpha_a \ , (\textbf{A}^{-1}\textbf{P}^T\textbf{x})_a  -t(\textbf{A}^{-1}\textbf{P}^T{\textbf{m}}^\prime)_a).
\end{align}
Plugging (\ref{dato0}) into (\ref{variaz0}), we get an expression for  $\Phi_{\nu}(t,\boldsymbol{x})$ and, setting $t=\beta$ and $\textbf{x}(\textbf{h}) =\beta  (\textbf{P}^T)^{-1} \textbf{A}\textbf{h}$, we get the variational principle for the free energy ($\ref{freedef}$)  
 \begin{align}
 \label{caccapuzza3}
f_{\nu}(\beta, \boldsymbol{\alpha},\boldsymbol{\xi},h)&= \Phi_{\nu}(\beta, \textbf{x}(\textbf{h}) )
=\displaystyle\sup_{{{\textbf{m}^\prime}}\in \mathbb{R}^{\nu\times P}} \biggl[ \frac{\beta(\textbf{m}^\prime)^2}{2} +\displaystyle\sum_{a=1}^{\nu}\alpha_a f(\beta  c k_a\alpha_a \ , \beta h_a \ - \beta(\textbf{A}^{-1}\textbf{P}^T{\textbf{m}}^\prime)_a\biggr].
\end{align}

Coming back  $\mathbf{m}^{\prime} \rightarrow \mathbf{m}$, we can rewrite (\ref{caccapuzza3}) as
\begin{align}
 \label{caccapuzza22}
 f_{\nu}(\beta, \alpha, h)&=\displaystyle\sup_{{\textbf{m}}\in \mathbb{R}^{\nu\times P}} \biggl( \frac{\beta(\textbf{m}, \textbf{J}^c\textbf{m})}{2} 
  +\displaystyle\sum_{a=1}^{\nu}\alpha_a f(\beta  ck_a\alpha_a \ , \beta h_a \ - (\beta(\textbf{A}^{-1}\textbf{J}^c{\textbf{m}^{\mu}})_a)\biggr).
\end{align}
\par\noindent
as stated in Theorem $\ref{THM}$.

%
By differentiating the argument of (\ref{caccapuzza22}) with respect to $\textbf{m}$ we get
\begin{equation}
\textbf{J}^c(\textbf{m}-\textbf{M})=0\ \   \implies\ \  \textbf{m}=\textbf{M}, 
\label{jottavia}
\end{equation}
since $\textbf{J}^c$ is positive definite and where
\begin{align}
\textbf{M}_a &=\nabla_{\boldsymbol{x}}f(\beta c k_a\a_a, \boldsymbol{x}) \mid_{\textbf{x}=\beta \textbf{h}_a - \beta(\textbf{A}^{-1}\textbf{J}^c\textbf{m})_{a}}.
\end{align}
Otherwise stated, $\boldsymbol{M}_a $ is the expected magnetization for a standard Hopfield model with size $N_a$, set at temperature $(\beta c k_a\a_a)^{-1}$ and with external fields $\beta \textbf{h}_a - \beta(\textbf{A}^{-1}\textbf{J}^c\textbf{m})_{a}$.
Recalling the self-consistent equation (\ref{gaiaa}) for the magnetization of the standard Hopfield model we get that $\boldsymbol{M}_a$ is solution of the following equation
\begin{align}
\boldsymbol{M}_a&=\biggl<\boldsymbol{\xi}\tanh\biggl(\beta ck_a\alpha_a \ \boldsymbol{\xi}\cdot\boldsymbol{M}_a+ \b \boldsymbol{\xi}\cdot\boldsymbol{h}_a -\beta\  \boldsymbol{\xi}\cdot \displaystyle(\textbf{A}^{-1}\textbf{J}^c\textbf{m})_{a}\biggr)\biggr>_{\boldsymbol{\xi}}
\nonumber\\
&=\biggl<\boldsymbol{\xi}\tanh\biggl(\beta ck_a\alpha_a \boldsymbol{\xi}\cdot(\boldsymbol{M}_a-\boldsymbol{m}_a)+ \b  \boldsymbol{\xi}\cdot  \displaystyle(\textbf{A}^{-1}\textbf{J}\textbf{m})_{a}   +\b\ \boldsymbol{\xi}\cdot\boldsymbol{h}_a\biggr)\biggr>_{\boldsymbol{\xi}}.
\end{align}
\par\noindent 
Thus from  (\ref{jottavia}) we get $(\ref{fiorellona})$.

\end{proof}


\section{Conclusions} \label{sec:conclusions}

In this work we considered multi-species Hopfield model in which different groups of neurons interact through a Hebbian-like coupling structure with different intensities.  The model is  completely solvable even in the non-convex region, where inter-groups couplings are stronger than those intra-groups. A variational principle for the free energy is introduced, whose solutions define also the possible thermodynamic states, described in terms of the Mattis magnetizations. The  strategy generalizes the one introduced in \cite{genovese2016non} and is based on a suitable Hamiltonian convexification, together with an interpolating procedure, that allows  the mapping with a mechanic problem, solved through an effective Hamilton-Jacobi equation.

The introduced model contains both the ingredients of the recent  Deep Boltzmann machines: non-convexity and Hebbian structure. Actually the latter results from the partial marginalization over some groups of neurons in a non-convex structure and is at the root of the Deep Boltzmann Machine ability to store patterns of informations, with different levels of correlations. In fact we have shown two special cases of the model: the BAM, a two layer neural networks with Hebbian interactions, able to store couples of patterns and a 3-layers RBM with hidden gaussian units. 

The analysis of this paper was performed in the  low load regime, i.e. when the Hebbian structure is composed of a sub-extensive (with respect to the system size) number of patterns. This facilitates a lot the analysis because only the ferromagnetic nature of the interaction counts and the  glassy states due to patterns interferences are not thermodynamic. The high load regime is still an open challenge: rigorous results for the standard (convex) Hopfield model lack, since rigorous results for the (non-convex) multi-species SK model are still incomplete.  Non-convexity appears in many problems of optimization in statistical inference, from community detection in networks to low rank matrix factorization 
\cite{rank1proof, lowrankmezard} 
thus an effort in this direction is necessary. 

Anyway one could start using non-rigorous statistical-physics techniques \cite{Ton,MPV} to study how the critical load of the standard Hopfield model changes in terms of the number of groups, their relative sizes,  the shape of the intensity matrix and, not last, the type of the neurons (Boolean, Gaussian, etc.), being the last closely related to the activation function used in practical applications.

\vspace{1cm}

{\bf Acknowledgements}\\
E.A.  acknowledges financial support by Sapienza Universit\`a di Roma (project no. RG11715C7CC31E3D).\\ 
D.T.\ is supported by Scuola Normale Superiore and National Group of Mathematical Physics GNFM-INdAM.

\newpage

\appendix

\section{Proof of Lemma 1}
\noindent Before proving Lemma \ref{lemma1} we define the matrix $\tilde{\textbf{J}}^c$ as
\begin{equation}
\tilde{\textbf{J}}^c= c \ \textbf{diag}(k^*\alpha_1^2, .. . . , k^*\alpha_\nu^2) -\textbf{J},
\label{minimum}
 \end{equation}
where $k^*$ is defined as in Lemma (\ref{lemma1}), then it holds the following
\begin{lemma} 
The quadratic form stemming from (\ref{minimum}) is positive definite if $c>1+\displaystyle\frac{\nu-1}{k^*}$. Moreover, we get {\normalfont{$\tilde{\textbf{J}}^c = \tilde{\textbf{P}}^T\tilde{\textbf{P}}$}}, where the rows of {\normalfont $\tilde{\textbf{P}}$} are the linearly independent vectors given by
\begin{align*}
& \textbf{v}^1= \frac{\sqrt{k^*(c-1)-(\nu-1)}}{\sqrt{\nu}} (\alpha_1, . . . \ , \alpha_\nu)
\\
&\textbf{v}^2= \frac{\sqrt{k^*(c-1)+1}}{\sqrt{2}} (\alpha_1, -\alpha_2, 0, . . . \ , 0)
\\
&  \vdots
\\
&{\textbf{v}}^a= \frac{\sqrt{k^*(c-1)+1}}{\sqrt{a(a-1)}} (\alpha_1, . . .. \ , \alpha_{a-1}, \alpha_{a}-a\alpha_{a}, 0, . . . ,0)
\\
& \vdots
\\
&\textbf{v}^\nu= \frac{\sqrt{k^*(c-1)+1}}{\sqrt{\nu(\nu-1)}} (\alpha_1, . . . \ , \alpha_{\nu-1}, \alpha_{\nu}-\nu\alpha_{\nu}).
\end{align*}
\label{lemma2}
\end{lemma}
\begin{proof}
Let $\textbf{A}=\textbf{diag}(\alpha_1, \dots , \alpha_{\nu})$
and  define
\begin{equation}
\tilde{\textbf{T}}_{ab}^c :=\begin{cases}
                & (c-1)k^* \; \; \; \text{se } \; \; a=b
                \\
                \\
                &  \  -1 \; \; \;  \ \ \  \ \  \text{se } \; \; a\neq b.
                \end{cases}
                \label{tc}
\end{equation}
such that we can write  $\tilde{\textbf{J}}^c=\textbf{A}\tilde{\textbf{T}}\textbf{A}$
We now derive the eigenvectors of $\tilde{\textbf{T}}^c$. In order for \textbf{w} to be eigenvector with eigenvalue $\lambda$ the following homogeneous system must be fulfilled
\begin{equation}
(\tilde{\textbf{T}}^c-\lambda\mathbb{I}_\nu)\textbf{w}=0.
\label{secolare}
\end{equation}
\par\noindent 
Focusing on the $i$-th component 
\begin{equation*}
((\tilde{\textbf{T}}^c-\lambda\mathbb{I}_\nu)\textbf{w})_{i}= -\displaystyle\sum_{a\neq i} w_{a} + ((c-1)k^*-\lambda)w_{i}= - \displaystyle\sum_{a=1}^{\nu} w_{a} + ((c-1)k^*-\lambda+1) w_{i}.
\end{equation*}
If $\displaystyle\sum_{a=1}^{\nu} w_{a}\neq 0$ then $\textbf{w}^1={\nu}^{-1/2}(1, \dots , 1)$  is eigenvector with eigenvalue $\lambda= (c-1)k^*+1-\nu$, being positive if  $c>1+(\nu-1)/k^*$.
\newline
The remaining eigenvectors live in the subspace of dimension $\nu-1$, characterized by $\displaystyle\sum_{a=1}^{\nu} w_{a}=0$ and  are related to the same eigenvalue $\lambda=(c-1)k^*+1$.
Normalizing, we get that the base of eigenvectors of the subspace orthogonal to the subspace $\mathbb{R} \textbf{w}^1$ is
\begin{align*}
& \textbf{w}^2=\frac{\sqrt{k^*(c-1)+1}}{\sqrt{2}}(1,-1, 0, \dots , 0)
\\
& \vdots
\\
&\textbf{w}^a=\frac{\sqrt{k^*(c-1)+1}}{\sqrt{a(a-1)}}(1,\dots , 1, 1-a, 0 , \dots , 0)
\\
& \vdots
\\
&\textbf{w}^a=\frac{\sqrt{k^*(c-1)+1}}{\sqrt{\nu(\nu-1)}}(1,\dots , 1, 1-\nu).
\end{align*}
Let ${\textbf{P}}^\prime$ the matrix of dimension $\nu \times \nu$, whose rows are the vectors $\textbf{w}_a$, $a=1, \dots, \nu$.
The matrix $\textbf{P}^\prime$ is such that ${\textbf{P}^\prime}^T\textbf{P}^\prime={\tilde{\textbf{T}}}^c$, therefore if we pose $\tilde{\textbf{P}}=\textbf{P}^\prime\textbf{A}$
we get
\begin{equation*}
\tilde{\textbf{P}}^T\tilde{\textbf{P}}= (\textbf{P}^\prime\textbf{A})^T(\textbf{P}^\prime\textbf{A})=\textbf{A}{\textbf{P}^\prime}^T\textbf{P}^\prime\textbf{A}= \textbf{A}{\tilde{\textbf{T}}}^c \textbf{A}=\tilde{\textbf{J}}^c.
\end{equation*}
Finally, we notice that the rows of $\tilde{\textbf{P}}$ are 
\begin{align*}
& \textbf{v}^1= \frac{\sqrt{(c-1)k^*-\nu+1}}{\sqrt{\nu}} (\alpha_1, . . . \ , \alpha_\nu)
\\
&\textbf{v}^2= \frac{\sqrt{k^*(c-1)+1}}{\sqrt{2}} (\alpha_1, -\alpha_2, 0, . . . \ , 0)
\\
& \vdots
\\
&{\textbf{v}}^a= \frac{\sqrt{k^*(c-1)+1}}{\sqrt{a(a-1)}} (\alpha_1, . . . \ , \alpha_{a-1}, \alpha_{a}-a\alpha_{a}, 0, . . . ,0)
\\
& \vdots
\\
&\textbf{v}^\nu= \frac{\sqrt{k^*(c-1)+1}}{\sqrt{\nu(\nu-1)}} (\alpha_1, . . . \ , \alpha_{\nu-1}, \alpha_{\nu}-\nu\alpha_{\nu}).
\end{align*}
\end{proof}
\par\noindent
We are now ready to prove Lemma (\ref{lemma1}).
\begin{proof}
We want to prove that if the condition $c>1+\displaystyle\frac{\nu-1}{k^*}$ holds, then (\ref{matj}) is positive definite.
Along the proof we will exploit the fact that (\ref{minimum}) is positive definite and we shall consider the diagonal matrix given by the difference between (\ref{matj}) and (\ref{minimum}):
\begin{equation}
\textbf{J}^c-\tilde{\textbf{J}}^c =\begin{cases}
                                                                                                  & (c-1)(k_a-k^*)\alpha_a^2 \; \; \; \text{se } \; \; a=b
                                                                                                   \\
                                                                                                   \\
                                                                                                  &  0\; \; \;\;\;\;  \ \ \  \ \  \text{se } \; \; a\neq b.      
                                                                                                   \end{cases}
\end{equation}
If $c>1+\displaystyle\frac{\nu-1}{k^*}$ then $\textbf{J}^c-\tilde{\textbf{J}}^c >0$. If $\textbf{u}\in\mathbb{R}^\nu$ is an arbitrary vector
\begin{equation}
\textbf{u}^T \bigl(\textbf{J}^c-\tilde{\textbf{J}}^c\bigr) \textbf{u}\ge 0
\end{equation} 
from which we get 
\begin{equation}
\textbf{u}^T \textbf{J}^c\textbf{u}> \textbf{u}^T\tilde{\textbf{J}}^c\textbf{u}> 0.
\end{equation} 
The last inequality is ensured by Lemma (\ref{lemma2}) and this concludes the proof.
\end{proof}


\begin{thebibliography}{9}

\bibitem{Bengio}
Y. Bengio, Y. LeCun, G. Hinton. Deep Learning. {\em Nature}, 521: 436-444 (2015).

\bibitem{DL-book} I. Goodfellow, Y. Bengio, A. Courville, {Deep Learning}, Google book (2016).

\bibitem{AMIT} D.J. Amit, {\em Modeling brain function: The world of attractor neural networks}, Cambridge University Press, (1992).

\bibitem{Ton} A.C.C. Coolen, R. K\"{u}hn, P. Sollich, {\em Theory of Neural Information Processing Systems}, Oxford Press, Oxford, (2005).

\bibitem{ags1} D.J. Amit, H. Gutfreund, H. Sompolinsky, {\em Spin Glass model of neural networks}, Phys. Rev. A \textbf{32}, 1007-1018,  (1985).

\bibitem{ags2} D.J. Amit, H. Gutfreund, H. Sompolinsky, {\em Storing infinite numbers of patterns in a spin glass model of neural networks}, Phys. Rev. Lett. \textbf{55}, 1530-1533, (1985).

\bibitem{hinton1} D.H. Hackley, G.E. Hinton, T.J. Sejnowski, {\em A learning alghoritm for Boltzmann machines}, Cognitive Science \textbf{9}(1):147, (1985).

\bibitem{hinton2} R. Salakhutdinov, G.E. Hinton, {\em Deep Boltzmann machines}, AISTATS \textbf{1}, 3 (2009).

\bibitem{RBM1} G.E. Hinton, S. Osindero, Y.W. Teh, {\em A fast algorithm for deep belief nets}, Neural Comp. \textbf{18}, 1527-1554, (2006).

\bibitem{RBM2} H. Larocelle, M. Mandel, R. Pascanu, Y. Bengio, {\em Learning algorithms for the classification restricted Boltzmann machine}. J. Mach. Learn. \textbf{13}, 643-669, (2012).

\bibitem{BBCS} A. Barra, A. Bernacchia, E. Santucci, P. Contucci,  {\em On the equivalence of Hopfield networks and Boltzmann machines}, Neur. Net. \textbf{34}:1-9 (2012).

\bibitem{BGST1} A. Barra, G. Genovese, P. Sollich and D. Tantari, {\em Phase transitions in Restricted Boltzmann Machines with generic priors}, Phys. Rev. E \textbf{96}(4):042156 (2017).

\bibitem{BGST2} A. Barra, G. Genovese, P. Sollich, D. Tantari, {\em Phase diagram of restricted Boltzmann machines and generalized Hopfield networks with arbitrary priors}, Phys. Rev. E \textbf{97}(2):022310 (2018).

\bibitem{remi} J. Tubiana, R. Monasson, {\em Emergence of Compositional Representations in Restricted Boltzmann Machines}, Phys. Rev. Lett. \textbf{118}, 138301 (2017).

\bibitem{huang} H. Huang, {\em Statistical mechanics of unsupervised feature learning in a restricted Boltzmann machine with binary synapses}, Journal of Statistical Mechanics: Theory and Experiment 2017(5):053302 (2017).

\bibitem{huang2} H. Huang, {\em Role of zero synapses in unsupervised feature learning}, J. Phys. A, \textbf{51}(8),08LT01 (2018)

\bibitem{Hebb} O.D. Hebb, {\em The organization of behaviour: a neuropsychological theory}, Pshyc. Press (1949).

\bibitem{kosko1988bidirectional}
B. Kosko, {\em Bidirectional associative memories}, IEEE Transactions on Systems, man, and Cybernetics, \textbf{18}(1):49--60 (1988).

\bibitem{kurchan1994statistical}
J.~Kurchan, L.~Peliti, and M.~Saber, {\em A statistical investigation of bidirectional associative memories ({BAM})}, Journal de Physique I, 4(11):1627--1639 (1994).

\bibitem{englisch1995bam}
H.~Englisch, V.~Mastropietro, B.~Tirozzi, {\em The {BAM} storage capacity}, Journal de Physique I, 5(1):85--96 (1995).

\bibitem{barra2015multi}
A.~Barra, P.~Contucci, E.~Mingione, D.~Tantari, {\em Multi-species mean field spin glasses. {Rigorous} results}, Annales Henri Poincar{\'e}, \textbf{16}, 691--708 (2015).
  
\bibitem{Bip} A. Barra, G. Genovese, F. Guerra, {\em Equilibrium statistical mechanics of bipartite spin systems}, J. Phys. A: Math. Theor. \textbf{44}, 245002 (2011).

\bibitem{barra2014mean}
A.~Barra, A.~Galluzzi, F.~Guerra, A.~Pizzoferrato, D.~Tantari, {\em Mean field bipartite spin models treated with mechanical techniques}, Eur. Phys. J. B, \textbf{87}(3):74, (2014).

\bibitem{panchenko2015free}
D.~Panchenko, {\em The free energy in a multi-species {Sherrington}-{Kirkpatrick} model}, The Annals of Probability, \textrm{43}(6):3494--3513 (2015).

\bibitem{genovsinc} G. Genovese, D. Tantari, {\em Overlap synchronisation in multipartite random energy models}, J. Stat. Phys., \textbf{169}(6), 1162-1170 (2017).

\bibitem{Contucci} P. Contucci, M.Fedele, C.Vernia, {\em Inverse problem robustness for multi-species mean field spin models}, J. Phys. A \textbf{46}, 065001 (2013).

\bibitem{genovese2016non}
G.~Genovese, D.~Tantari, {\em Non-convex multipartite ferromagnets}, J. Stat. Phys. \textbf{163}(3):492--513 (2016).

\bibitem{article}
E.~Agliari, A.~Barra, A.~Galluzzi, D.~Tantari, F.~Tavani, {\em A walk in the statistical mechanical formulation of neural networks -
  {Alternative} routes to {Hebb} prescription}, NCTA2014: Neural computation theory \& application, \textbf{7}, 210--217 (2014).

\bibitem{MC-P} W.S. McCulloch, W. Pitts, {\em A logical calculus of the ideas immanent in nervous activity},  Bull. Math. Biophys. \textbf{5}, 115-133, (1943).

\bibitem{correlated} E.J. Gardner, D. J. Wallace, N. Stroud, {\em Training with noise and the storage of correlated patterns in a neural network model}, J. Phys. A \textbf{22}(12):2019, (1989).
      
\bibitem{ABDG} E. Agliari, A. Barra, A. De Antoni, A. Galluzzi, {\em Parallel retrieval of correlated patterns: From Hopfield networks to Boltzmann machines}, Neural Networks \textbf{38}:52-63 (2013).

\bibitem{gutfreund1988neural}
H.~Gutfreund, {\em Neural networks with hierarchically correlated patterns}, Phys. Rev. A, \textbf{37}(2):570 (1988).

\bibitem{prlnoi1} E. Agliari, A. Barra, A. Galluzzi, F. Guerra, F. Moauro, {\em Multitasking associative networks}, Phys. Rev. Lett. \textbf{109}, 268101, (2012).

\bibitem{multi} P. Sollich, D. Tantari, A. Annibale, A. Barra, {\em Extensive parallel processing on scale free networks}, Phys. Rev. Lett. \textbf{113}, 238106, (2014).   

\bibitem{immhigh} 
E .Agliari, A. Annibale, A. Barra, A.C.C. Coolen, D. Tantari, {\em Immune networks: multitasking capabilities near saturation}, J. Phys. A, \textbf{46}, 415003 (2013).

\bibitem{immmed} E. Agliari, A. Annibale, A. Barra, A.C.C. Coolen, D. Tantari, {\em Immune networks: multi-tasking capabilities at medium load}, J. Phys. A, \textbf{46}, 335101 (2013).

\bibitem{immlett} 
E. Agliari, A. Annibale, A. Barra, A.C.C. Coolen, D. Tantari, {\em Retrieving infinite numbers of patterns in a spin-glass model of immune networks}, Europhys. Let. \textbf{117} (2), 28003 (2017).

\bibitem{agliari2014multitasking}
E.~Agliari, A.~Barra, A.~Galluzzi, and M.~Isopi, {\em Multitasking attractor networks with neuronal threshold noise}, Neur. Net., \textbf{49}:19--29 (2014).

\bibitem{BGG-JSP2010} 
A. Barra, G. Genovese, F. Guerra, {\em The replica symmetric approximation of the analogical neural network}, J. Stat. Phys. \textbf{140}(4):784-796 (2010).

\bibitem{barra2012glassy}
A.~Barra, G.~Genovese, F.~Guerra, D.~Tantari, {\em How glassy are neural networks?}, J. Stat. Mech.,
\textbf{2012}(07):P07009 (2012).
  
\bibitem{bg} A. Barra, F. Guerra, {\em About the ergodic regime in the analogical Hopfield neural networks: Moments of the partition function}, J. Math. Phys. 49, 125217, (2008).

\bibitem{gaussSK}
A. Barra, G. Genovese, F. Guerra, D. Tantari, {\em About a solvable mean field model of a Gaussian spin glass}, 
J. Phys. A, \textbf{47}(15):155002 (2014). 

\bibitem{legendre} 
G. Genovese, D. Tantari, {\em Legendre Duality of Spherical and Gaussian Spin Glasses},
Mathematical Physics, Analysis and Geometry, \textbf{18}:10 (2015).

\bibitem{anergy} E. Agliari, A. Barra, G. Del Ferraro, F. Guerra, D. Tantari, 
{\em Anergy in self-directed B lymphocytes: A statistical mechanics perspective}, J. Theor. Biol. \textbf{375}, 21-31 (2015).

 \bibitem{sompodil} H. Sompolinsky, Phys. Rev. A, {\em Neural networks with nonlinear synapses and a static noise}, \textbf{34}, 2571(R) (1986).
 
\bibitem{wemmedil} B. Wemmenhove, A. C. C. Coolen, {\em Finite connectivity attractor neural networks}, J. Phys. A \textbf{36}, 9617 (2003).

\bibitem{prlnoi3} E. Agliari, A. Barra, A. Galluzzi, F. Guerra, D. Tantari, F. Tavani, {\em Retrieval capabilities of hierarchical networks: From Dyson to Hopfield}, Phys. Rev. Lett. \textbf{114}, 028103, (2015).

\bibitem{dysnn}
E. Agliari, A. Barra, A. Galluzzi, F. Guerra, D. Tantari, F. Tavani, {\em Hierarchical neural networks perform both serial and parallel processing}, Neur. Net. \textbf{66}, 22-35 (2015).

\bibitem{dysjpa} 
E. Agliari, A. Barra, A. Galluzzi, F. Guerra, D. Tantari, F. Tavani, {\em Metastable states in the hierarchical Dyson model drive parallel processing in the hierarchical Hopfield network}, J. Phys. A, \textbf{48}(1):015001 (2014).

\bibitem{dyspre} 
E. Agliari, A. Barra, A. Galluzzi, F. Guerra, D. Tantari, F. Tavani, {\em Topological properties of hierarchical networks}, Phys. Rev. E \textbf{91}(6), 062807 (2015).
    
\bibitem{pastur1994replica}
L. Pastur, M. Shcherbina, B.Tirozzi, {\em The replica-symmetric solution without replica trick for the {Hopfield} model}, J. Stat. Phys., \textbf{74}(5):1161--1183 (1994).

\bibitem{Tala1}  M. Talagrand, {\em Rigorous results for the Hopfield model with many patterns}, Prob. Theory and Rel. Fields, \textbf{110}(2):177-275 (1998).

\bibitem{Tala2} M. Talagrand, {\em Exponential inequalities and convergence of moments in the replica-symmetric regime of the Hopfield model}, Ann. Prob. \textbf{28}(4):1393-1469, (2000).

\bibitem{BGP3} A. Bovier, V. Gayrard, P. Picco, {\em Gibbs states of the Hopfield model with extensively many patterns}, J. Stat. Phys. \textbf{79}, 395-414 (1995).

\bibitem{BG5} A. Bovier, V. Gayrard, {\em The retrieval phase of the Hopfield model, A rigorous analysis of the overlap distribution}, Prob. Theor. Rel. Fields \textbf{107}, 61-98 (1995).

\bibitem{bov-gen} A. Bovier, V. Gayrard, {\em Hopfield models as generalized random mean field models}, Progress in Probability Vol. 41, (A. Bovier and P. Picco, Editors) Birkauser, Boston, (1997).

\bibitem{barra0} A. Barra, {\em The mean field Ising model trough interpolating techniques}, J. Stat. Phys. \textbf{132}(5):787-809 (2008)

\bibitem{guerra-HJ} F. Guerra, {\em Sum rules for the free energy in the mean field spin glass model}, Fields Inst. Comm. \textbf{30}:161 (2001).


\bibitem{Liao}
X. Liao, J. Yu. Qualitative analysis of Bi-directional Associative Memory with time delay. {\em Int. J. Circ. Theor. Appl.}, 26(3): 219-229 (1998).

\bibitem{Cao}
J. Cao, M. Xiao, {\em Stability and Hopf Bifurcation in a Simplified BAM Neural Network With Two Time Delays}, IEEE Transactions on Neural Networks, 18(2): 416-430 (2007).
 
\bibitem{Cao2}
J. Cao, L. Wang, {\em Exponential stability and periodic oscillatory solution in BAM networks with delays}, IEEE Transactions on Neural Networks, 13(2): 457-463 (2007).

\bibitem{Cao3}
J. Cao, {\em Global asymptotic stability of delayed bi-directional associative memory neural networks}, Applied Mathematics and Computation, 142(2-3): 333-339 (2003).

\bibitem{Cao4}
J. Cao, Y. Wan, {\em Matrix measure strategies for stability and synchronization of inertial BAM neural network with time delays}, Neur. Net., \textbf{53}: 165-172 (2014).

\bibitem{Park}
J.H. Park, C.H.Park, O.M. Kwon, S.M. Leed, {\em A new stability criterion for bidirectional associative memory neural networks of neutral-type}, Applied Mathematics and Computation, \textbf{199}(2):716-722 (2008).

\bibitem{mezard} M. Mezard, {\em Mean-field message-passing equations in the Hopfield model and its generalizations},  Phys. Rev. E \textbf{95}(2), 022117 (2017).
 
 \bibitem{gabrie1} M. Gabri\'e, E. W. Tramel, F. Krzakala, {\em Training restricted Boltzmann machine via the Thouless-Anderson-Palmer free energy}, Advances in Neural Information Processing Systems \textbf{1}, 640-648, (2015).

\bibitem{HJ-Jstat} A. Barra, A. Di Biasio, F. Guerra, {\em Replica symmetry breaking in mean-field spin glasses through the Hamilton Jacobi technique}, J. Stat. Mech. 2010(09), P09006, (2010).

\bibitem{barra2013mean}
A.~Barra, G.~Dal~Ferraro, D.~Tantari, {\em Mean field spin glasses treated with {PDE} techniques}, Eur. Phys. J. B, \textbf{86}(7):332 (2013).

\bibitem{genovese2009mechanical}
G. Genovese, A. Barra, {\em A mechanical approach to mean field spin models}, J. Math. Phys., \textbf{50}(5):053303 (2009).

\bibitem{evans1998partial}
L. Evans, {\em Partial differential equations (graduate studies in mathematics}, vol. 19, American Mathematical Society, Providence (1998).

\bibitem{lax} P. Cannarsa, C. Sinestrari, {\em Semiconcave functions}, Hamilton-Jacobi equations, and optimal control, Boston, Birkhauser (2004).
    


\bibitem{rank1proof} 
J. Barbier, M. Dia, N. Macris, F. Krzakala, T. Lesieur, L. Zdeborova, {\em Mutual information for symmetric rank-one matrix estimation: A proof of the replica formula} Advances in Neural Information Processing Systems, 424-432, (2015).

\bibitem{lowrankmezard} Y. Kabashima, F. Krzakala, M. M\`ezard, A. Sakata, L. Zdeborova, {\em Phase transitions and sample complexity in Bayes-optimal matrix factorization}, IEEE Transactions on Information Theory 62 (7), 4228-4265 (2016).    

\bibitem{MPV} M. M\'ezard, G. Parisi, M. A. Virasoro, {\em Spin glass theory and beyond}, World Scientific, Singapore, (1987).





\end{thebibliography}
\end{document}